\documentclass{article}
\usepackage{times}
\usepackage{soul}
\usepackage{url}
\usepackage[hidelinks]{hyperref}
\usepackage[utf8]{inputenc}
\usepackage[small]{caption}
\usepackage{graphicx}
\usepackage{booktabs}
\usepackage{algorithm}
\usepackage{amsthm,amsmath}
\usepackage{mathtools}
\usepackage{microtype}
\usepackage{cleveref}
\usepackage{nimbusmononarrow} 

\theoremstyle{plain}
\newtheorem{theorem}{Theorem}
\newtheorem{corollary}[theorem]{Corollary}
\newtheorem{lemma}[theorem]{Lemma}

\theoremstyle{definition}
\newtheorem{definition}[theorem]{Definition}
\newtheorem{example}[theorem]{Example}
\theoremstyle{remark}

\newtheorem*{theorem-repeat}{Restatement}

\usepackage{times}
\usepackage{soul}
\usepackage{url}
\usepackage[hidelinks]{hyperref}
\usepackage[utf8]{inputenc}
\usepackage{caption}
\usepackage{graphicx}
\usepackage{amsmath}
\usepackage{booktabs}
\usepackage{algorithm}
\usepackage{algorithmic}
\usepackage[british]{babel}
\usepackage[usenames,x11names]{xcolor}
\usepackage{relsize,setspace}
\usepackage{stmaryrd}
\usepackage[font=it]{caption}
\usepackage{tikz,pgf}\usetikzlibrary{shapes,backgrounds}

\usepackage{amssymb}
\usepackage{pifont}

\usepackage{cleveref}  

\usepackage{multirow}
\usepackage{comment}
\usepackage{xspace}
\usepackage{pdflscape}
\usepackage{yhmath}

\usepackage[inline]{enumitem}

\setlist{leftmargin=*,noitemsep,parsep=0.2ex,topsep=0.3ex}
\usepackage{scalerel}

\usepackage{microtype}

\usepackage{ifthen}
\newif\iflong
\longfalse

\usepackage{todonotes}

\newcommand{\FormatArithmeticComplexityClass}[1]{\ensuremath{\textsc{#1}}\xspace}

\newcommand{\ExpTime}{\FormatArithmeticComplexityClass{ExpTime}}
\newcommand{\NExpTime}{\FormatArithmeticComplexityClass{NExpTime}}

\newcommand{\pred}[1]{\mbox{\tt #1}}

\newcommand{\spform}[1]{\mathsf{#1}}

\newcommand{\mathcom}[3]{ \newcommand{#1}[#2]{\mbox{$#3$}}}
\mathcom{\imp}{0}{\ \rightarrow\ }            
\mathcom{\rimp}{0}{\ \leftarrow\ }            

\mathcom{\con}{0}{\ \wedge\ }                 
\mathcom{\dis}{0}{\ \vee\ }                   
\mathcom{\n}{0}{\neg}                     
\mathcom{\dimp}{0}{\ \leftrightarrow\ }       
\mathcom{\corresponds}{0}{\ \Lleftarrow\! \! \Rrightarrow\ }
\mathcom{\A}{0}{\forall}                  
\mathcom{\E}{0}{\exists}     

\def\Box{\mathop\square}
\def\Diamond{\mathop\lozenge}

\mathcom{\tuple}{1}{\langle #1 \rangle}

\def\eqdef{\mathrel{\ =_{\mbox{\em \tiny def}}\ }}

\def\PredStandpointLogic{{\mathbb{S}_{\scaleobj{0.85}{\mathrm{FO}}}}}

\def\FOformulas{\PredStandpointLogic}

\def\st{\hbox{$\spform{s}$}\xspace}

\def\pr{\pi}

\def\Precs{\Pi}

\def\standb#1{\Box\nolimits_{\spform{#1}}}
\def\standd#1{\Diamond\nolimits_{\spform{#1}}}

\def\standbx#1{\Box\nolimits_{\scaleobj{0.8}{\spform{#1}}}}
\def\standdx#1{\Diamond\nolimits_{\scaleobj{0.8}{\spform{#1}}}}

\def\standbs{\standb{s}}

\def\standde{\standd{e}}
\def\standds{\standd{s}}

\def\allstandb{\standb{*}}
\def\allstandd{\standd{*}}

\xspace
\def\f\xspacestandtopre{\hbox{$\sigma\,$}\xspace}
\def\fpretov\xspacealue{\hbox{$\delta\,$}\xspace}

\def\ModSat#1||-#2{#1\models #2}
\def\NotModSat#1||-#2{#1\nvDash #2}

\newcommand{\skipit}[1]{} 
\newcommand{\addit}[1]{} 

\newcommand{\define}[1]{\emph{#1}}
\newcommand{\N}{\mathbb{N}}
\newcommand{\set}[1]{\left\{#1\right\}}

\renewcommand{\eqdef}{\mathrel{\,:=\,}}

\newcommand{\ebnfeq}{\mathrel{::=}}

\newcommand{\limplies}{\rightarrow}
\renewcommand{\land}{\mathrel{\wedge}}
\renewcommand{\lor}{\mathrel{\vee}}

\newcommand{\trans}{\mathrm{trans}}

\newcommand{\Sub}{\mathit{Sub}}

\def\Stands{\mathbf{S}}
\def\sts{\spform{s}} 
\def\Preds{\mathbf{P}}

\def\Consts{\mathbf{C}}
\def\Vars{\mathbf{V}}
\def\Terms{\mathbf{T}}
\def\E{\mathbf{E}}
\def\StandExps{\E_{\Stands}}
\def\ste{\spform{e}} 
\def\sigmaE{\sigma_{\E}}
\def\transE{\trans_{\E}}

\def\kstruct{\mathfrak{M}}

\def\struct{\mathcal{I}}
\def\varassign{v}
\def\intf{\cdot^{\struct}}
\def\Dom{\Delta}
\def\de{\delta}
\newcommand{\interpret}[2]{#1^{#2}}
\newcommand{\interprets}[1]{\interpret{#1}{\struct}}

\def\EL{\ensuremath{\mathcal{E\!L}}\xspace}

\def\FOSL{FOSL\xspace}

\def\foss{first-or\-der stand\-point structure\xspace}
\def\fosss{first-or\-der stand\-point structures\xspace}

\newenvironment{narrowalign}{\\[3pt]\mbox{}\hfill}{\hfill\mbox{}\\[4pt]}

\newcommand{\SROIQ}{\ensuremath{\mathcal{SROIQ}}\xspace}
\newcommand{\SROIQbs}{\ensuremath{\mathcal{SROIQ}b_s}\xspace}

\newcommand{\SHIQ}{$\mathcal{SHIQ}$\xspace}

\newcommand{\ELp}{$\mathcal{EL}{\mathord{+}}$\xspace}

\newcommand{\TwoExpTime}{\mbox{\sc{2\ExpTime}}\xspace}
\newcommand{\NTwoExpTime}{\mbox{\sc{N\TwoExpTime}}\xspace}

\newcommand{\Inter}{\mathcal{I}} 

\RequirePackage{graphicx}

\newcommand{\atleast}[1]{\mathord{\geqslant}#1\,}
\newcommand{\atmost}[1]{\mathord{\leqslant}#1\,}

\newcommand{\conc}[1]{#1}
\newcommand{\rol}[1]{#1}

\newcommand{\rolexpR}{\rol{R}}
\newcommand{\rolexpS}{\rol{S}}
\newcommand{\rolR}{\mathtt{R}}
\newcommand{\rolS}{\mathtt{S}}

\newcommand{\rolU}{\mathtt{U}}

\newcommand{\conA}{\mathtt{A}}

\newcommand{\conC}{\conc{C}}
\newcommand{\conD}{\conc{D}}

\def\modelPR{\langle \kstruct \rangle_{\Preds_{\mathtt{E}}}}
\def\rigidp{\mathtt{E}}

\def\num{num}
\def\dom{dom}
\def\stacked{\mathsf{stacked}}

\newcommand{\Ctwo}{\ensuremath{\mathcal{C}^2}}

\newcommand{\defend}{\hfill$\Diamond$}

\usepackage{kr}

\crefname{lemma}{lemma}{lemmas}   
\Crefname{lemma}{Lemma}{Lemmas}
\crefname{figure}{Fig.\!}{}
\Crefname{figure}{Figure}{}

\setlist[description]{leftmargin=1em, itemindent=-0.5em}

\title{Putting Perspective into OWL [sic]:\\Complexity-Neutral Standpoint Reasoning for 
Ontology Languages\\ 
via Monodic S5 over Counting Two-Variable First-Order Logic 
\\ \textnormal{(Extended Version with Appendix)} 
}

\author{%
	Lucía Gómez Álvarez$^{1}$\and
	Sebastian Rudolph$^{2}$\\
	\affiliations 
        \normalsize 
	$^1$INRIA, Université Grenoble Alpes\\
	$^2$TU Dresden and 
    Center for Scalable Data Analytics and Artificial Intelligence Dresden/Leipzig\\
	\emails 
        \normalsize 
	lucia.gomez-alvarez@inria.fr,
	sebastian.rudolph@tu-dresden.de
}

\begin{document}

\pagenumbering{arabic}
\maketitle

\begin{abstract}

Standpoint extensions of KR formalisms have been recently introduced to incorporate multi-perspective modelling and reasoning capabilities.
In such modal extensions, the integration of conceptual modelling and perspective annotations can be more or less tight, with monodic standpoint extensions striking a good balance as they enable advanced modelling while preserving good reasoning complexities.

We consider the extension of \Ctwo -- the counting two-variable fragment of first-order logic -- by monodic standpoints. At the core of our treatise is a polytime translation of formulae in said formalism into standpoint-free \Ctwo, requiring elaborate model-theoretic arguments. By virtue of this translation, the \NExpTime-complete complexity of checking satisfiability in \Ctwo{} carries over to our formalism. As our formalism subsumes monodic S5 over \Ctwo, our result also significantly advances the state of the art in research on first-order modal logics.

As a practical consequence, the very expressive description logics $\mathcal{SHOIQB}_s$ and $\mathcal{SROIQB}_s$ which subsume the popular W3C-standardized OWL~1 and OWL~2 ontology languages, are shown to allow for monodic standpoint extensions without any increase of standard reasoning complexity.

We prove that \NExpTime-hardness already occurs in much less expressive DLs as long as they feature both nominals and monodic standpoints.
We also show that, with inverses, functionality, and nominals present,  minimally lifting the monodicity restriction leads to undecidability. 
\end{abstract}

\section{Introduction}\label{sec:introduction}



Integrating knowledge from diverse, independently developed sources is a central problem in knowledge representation, particularly given the proliferation of available ontologies and other knowledge sources. Many of these ontologies -- often expressed in W3C-standardized dialects of the Web Ontology Language (OWL) \cite{owl2-overview} -- cover overlapping domains but embody varying conceptual frameworks and modelling choices. For instance, in the biomedical domain, some ontology ($\mathcal{O}_{\mathsf{Process}}$) might define $\pred{Tumour}$ as a dynamic biological process, whereas another ($\mathcal{O}_{\mathsf{Tissue}}$) might view it as an abnormal tissue structure. While the description logics (DLs) \cite{BHLSintroDL,RudolphDLfoundations} underpinning OWL are well-suited to coherently model a domain, they lack mechanisms for managing heterogeneous or conflicting perspectives, leading to notorious challenges in integration.

Standpoint logic (SL) \cite{gomez2021standpoint} is a recently proposed modal logic framework for multi-perspective reasoning and ontology integration. In a similar vein to epistemic logic, propositions with labelled modal operators $\standbs\phi$ and $\standds\phi$ express information relative to the \emph{standpoint} $\st$ and read, respectively: ``according to $\st$, it is \emph{unequivocal/conceivable} that $\phi$''. For instance, the formula $\standbx{Process}[\standdx{Tissue}[\pred{Tumour}]\sqsubseteq {=}1\pred{TriggeredBy.}\pred{Tumour}]$ expresses that, according to the $\mathsf{Process}$ standpoint, it is unequivocal that everything that is conceivably a $\pred{Tumour}$ from the $\mathsf{Tissue}$ standpoint has been triggered by exactly one $\pred{Tumour}$ (process). Similarly, 
$\standbx{Tissue}[\{\pred{patient1}\} \sqsubseteq \exists\pred{HasBodyPart.}(\pred{Tumour}\sqcap\{\pred{t1}\})]$ states that according to the $\mathsf{Tissue}$ standpoint, it is unequivocal that $\pred{patient1}$ has the $\pred{Tumour}$ $\pred{t1}$ as a body part. From both, we infer that according to the $\mathsf{Process}$ standpoint, $\pred{t1}$ was triggered by one $\pred{Tumour}$. Natural reasoning tasks over multi-standpoint specifications include gathering undisputed knowledge, determining knowledge that is relative to certain standpoints, and contrasting the knowledge from different standpoints.

The SL framework has promising applications in ontology integration, particularly in facilitating the interoperability of ontologies developed in isolation. For this reason, recent work has explored how it can be combined with logic-based formalisms underpinning the OWL family -- most notably with the DLs \EL \cite{monodicEL}, \ELp \cite{monodicELplus} and \SHIQ \cite{monodicSHIQ}. 
It has been shown that monodic extensions\footnote{Monodic extensions of first-order modal logic allow for one free variable in the scope of the modal operator, and for modalised axioms and concept expressions in the case of modal DLs.} of these languages with SL preserve the complexity of the standpoint-free DL, showing that joint reasoning over the integrated combination of possibly many ontologies is not fundamentally harder than reasoning with the ontologies in separation. 

Hitherto, an open question has been whether the same holds for the very expressive side of modelling languages, in particular DLs that would fully cover high-end contemporary ontology languages such as OWL~2~DL. The results obtained so far for such languages only considered \emph{sentential fragments} \cite{sententialFOLandOWL}, which is an easier but much more restricted case where there is no interplay between quantification and modal operators (in DLs, the modal operators would only occur on the axiom-level).

In this paper, we address this open question by considering the extension of \Ctwo -- the counting two-variable fragment of first-order logic, which in fact has already gained some popularity for serving as a logic to embed very expressive DLs into -- by monodic standpoints. After the preliminaries (in \Cref{sec:preliminaries}), we provide, in \Cref{sec:sat-in-monodic-c2}, a polytime translation of formulae in said formalism into plain \Ctwo\!, using elaborate model-theoretic arguments. From this, we establish that the \NExpTime-completeness of checking (finite) satisfiability in \Ctwo{} carries over to \emph{monodic standpoint \Ctwo{}}\!. 
As our formalism subsumes monodic S5 over \Ctwo\!, our result also significantly advances the state of the art in first-order modal logic.

\Cref{sec:application-ontology-languages} exposes how, as a consequence, the very expressive DLs $\mathcal{SHOIQB}_s$ and $\mathcal{SROIQB}_s$ which subsume the OWL~1 and OWL~2 ontology languages, also allow for monodic standpoint extensions without any increase of standard reasoning complexity. Moreover, in \Cref{sec:features-causing-trouble} we prove that \NExpTime-hardness already occurs in much less expressive DLs as long as they feature both nominals and monodic standpoints. Additionally, with inverses, functionality, and nominals present, minimally lifting the monodicity restriction by allowing for one distinguished rigid binary predicate leads to undecidability. Finally, the full proofs for most sketches can be found in the appendix.

\section{Preliminaries}\label{sec:preliminaries}

\subsection{First-Order Standpoint Logic}\label{sec:semantics}
We introduce syntax and semantics of first-order standpoint logic  (\FOSL, see \citeauthor{sententialFOLandOWL} \citeyear{sententialFOLandOWL}).
\begin{definition}\label{def:FOSLsyntax}
	The syntax of any \FOSL formula is based on a set $\Vars$ of \define{variables}, typically denoted with  $x,y,\ldots$, and a \define{signature} $\tuple{\Preds, \Consts, \Stands}$, consisting of \define{predicates} $\Preds$ (each associated with an arity \mbox{$n\in\N$}), \define{constants} $\Consts$ and \define{standpoint symbols} $\Stands$, usually denoted $\sts,\sts'$. In particular, $\Stands$ also contains $*$, the \emph{universal standpoint}. $\Vars$, $\Preds$, $\Consts$, and $\Stands$ are assumed to be pairwise disjoint.
	The set $\Terms$ of \define{terms} contains all constants and variables, that is, \mbox{$\Terms=\Consts\cup\Vars$}.

	The set $\StandExps$ of \define{standpoint expressions} is defined by
	$$\ste_1,\ste_2 \ebnfeq \sts \mid \ste_1\cap\ste_2 \mid \ste_1\cup\ste_2 \mid \ste_1\setminus\ste_2,$$
	with $\sts \in \Stands$. The set $\FOformulas$ of \FOSL \define{formulae} is then given by
	$$ \phi,\psi \ebnfeq \pred{P}(t_1,...\,,t_k) \mid t_1 \dot{=} t_2 \mid \neg\phi \mid \phi\,{\land}\,\psi \mid \exists^{\lhd n}  x.\phi \mid \standde\phi, $$
	where \mbox{$\pred{P}\in\Preds$} is a $k$-ary predicate; \mbox{$t_1,\dots,t_k\in\Terms$} are terms; \mbox{$\lhd$ is any of $\leq$, $=$, or $\geq$}; $n \in \mathbb{N}$; 
	\mbox{$x\in\Vars$}; and \mbox{$\ste\in\StandExps$}. \defend
\end{definition}

For a formula $\phi$, we denote the set of all of its subformulae by $\Sub(\phi)$.
The \define{size} of a formula is \mbox{$|\phi| \eqdef |\Sub(\phi)|$}.
The connectives and operators $\mathbf{t}$, $\mathbf{f}$, \mbox{$\phi\lor\psi$}, \mbox{$\phi\limplies\psi$}, \mbox{$\phi\leftrightarrow\psi$}, \mbox{$\forall x.\phi$}, and \mbox{$\standb{\ste}\phi$} are introduced as syntactic macros as usual -- in particular, $\forall x.\phi$ is used to abbreviate $\exists^{= 0} x.\neg\phi$.
In line with intuition, we may just write $\exists x.\phi$ instead of $\exists^{\geq 1}x.\phi$. We note that in full first-order logic, the somewhat exotic \emph{counting quantifiers} $\exists^{\lhd n}$ do not add extra expressivity compared to the non-counting ones, but they do make a difference when the number of variables is restricted. As this is where we are heading, it is convenient to start from this syntax definition.

\pagebreak
A first-order standpoint logic formula $\phi$ is called 
\begin{itemize}
\item \emph{monodic} if in each of its subformulae of the shape $\standde\psi$, the formula $\psi$ has at most one free variable,
\item \emph{$\mathcal{C}^2$} if it only uses the two variables $x$ and $y$ and predicates of arity ${\leq} 2$,
and 
\emph{plain $\mathcal{C}^2$} if it is $\mathcal{C}^2$ and does not use $\Diamond$,
\item \emph{S5} if the only standpoint expression used is $*$,  
\item \emph{nullary-free} if it does not use predicates of arity zero,
\item \emph{constant-free} if it does not use constants.
\end{itemize}

Moreover, we will call formulae of the form $\allstandd\phi$ \emph{monodic modal formulae} if they have one free variable and \emph{sentential modal formulae} if they have no free variables.


\begin{definition}\label[definition]{def:semantics-model}
	Given a signature $\tuple{\Preds,\Consts,\Stands}$, a \define{(first-order) standpoint structure} $\kstruct$ is a tuple $\tuple{\Dom, \Precs, \sigma, \gamma}$ where:
	\begin{itemize}
		\item $\Dom$ is a non-empty set, the \define{domain} of $\kstruct$;
		\item $\Precs$ is a non-empty set, called \define{precisifications} or \emph{worlds};
		\item $\sigma$ is a function mapping each standpoint symbol from $\Stands$ to a set of worlds (i.e., a subset of $\Precs$), with $\sigma(*)=\Precs$ fixed;
		\item $\gamma$ is a function mapping each precisification from $\Precs$ to an ordinary first-order structure $\struct$ over the domain $\Delta$, whose interpretation function $\intf$ maps\/:
		      \begin{itemize}
			      \item every predicate symbol $\pred{P}{\,\in\,}\Preds$ of arity $k$ to a $k$-ary relation \mbox{$\interprets{\pred{P}} {\,\subseteq\,} \Dom^k$},
			      \item each constant symbol $\mathtt{a}{\,\in\,}\Consts$ to a domain element \mbox{$\interprets{\mathtt{a}}{\in\,}\Dom$}.
		      \end{itemize}
		      For any two $\pr_1,\pr_2\in\Precs$ and every $\mathtt{a}\in\Consts$ we require $\interpret{\mathtt{a}}{\gamma(\pr_1)}=\interpret{\mathtt{a}}{\gamma(\pr_2)}$ (i.e., we enforce rigid constants). \defend
	\end{itemize}
\end{definition}


If in $\kstruct$, some predicate $\pred{P} \in \Preds$ satisfies $\interpret{\mathtt{P}}{\gamma(\pr_1)}=\interpret{\mathtt{P}}{\gamma(\pr_2)}$ for every $\pr_1,\pr_2\in\Precs$, we say that $\pred{P}$ is \emph{rigid} (in $\kstruct$) and allow ourselves to write $\interpret{\mathtt{P}}{\kstruct}$ instead of $\interpret{\mathtt{P}}{\gamma(\pr_1)}$

\begin{definition}\label[definition]{def:semantics}
	Let $\kstruct=\tuple{\Dom,\Precs,\sigma,\gamma}$ be a \foss for the signature $\tuple{\Preds,\Consts,\Stands}$ and $\Vars$ be a set of variables.
	%
	%
	A \define{variable assignment} is a function $\varassign:\Vars\to\Dom$ mapping variables to domain elements.
	Given a variable assignment $v$, we denote by $\varassign_{\set{x\mapsto\de}}$ the function mapping $x$ to $\de\in\Dom$ and any other variable $y$ to $\varassign(y)$.

	An interpretation function $\intf$ and a variable assignment specify how to interpret terms by domain elements\/:
	\iflong
\else
		We let $\interpret{t}{\struct,\varassign} = \varassign(x)$ if $t=x\in\Vars$, and $\interpret{t}{\struct,\varassign} = \interprets{a}$ if $t=a\in\Consts$.

	\fi
	\iflong
\else
		To interpret standpoint expressions, we lift $\sigma$ from $\Stands$ to all of $\StandExps$ via
        $\sigma(\ste_1\bowtie\ste_2) = \sigma(\ste_1)\bowtie\sigma(\ste_2)$ for ${\bowtie} \in \{\cup, \cap, \setminus\}$.
	\fi

    The satisfaction relation for formulae is defined in the usual way via structural induction.
	In what follows, let $\pr\in\Precs$ and let $\varassign:\Vars\to\Dom$ be a variable assignment;
	we now establish the definition of the satisfaction relation $\models$ for FOSL using pointed \fosss:\skipit{\footnote{We implicitly define $\phi$ to be (globally) \emph{satisfiable} iff \mbox{$\kstruct\models\phi$} for some structure $\kstruct$.
			Local satisfiability (satisfiability in some precisification) can be emulated via global satisfiability of $\allstandd\mkern-2mu\phi$.}}%
	\begin{align*}
		 & \kstruct,\!\pr,\!\varassign \models \pred{P}(t_1,\ldots,t_k)\!\! & \text{iff}\ \  & (\interpret{t_1}{\gamma(\pr), \varassign},\ldots,\interpret{t_k}{\gamma(\pr),\varassign}) \in \interpret{\pred{P}}{\gamma(\pr)} \\
		 & \kstruct,\!\pr,\!\varassign \models t_1 \dot{=}\, t_2 & \text{iff}\ \  & \interpret{t_1}{\gamma(\pr), \varassign} = \interpret{t_2}{\gamma(\pr),\varassign} \\
		 & \kstruct,\!\pr,\!\varassign \models \neg\phi                     & \text{iff}\ \  & \kstruct,\!\pr,\!\varassign\not\models\phi                                                                                      \\
		 & \kstruct,\!\pr,\!\varassign \models \phi\land\psi                & \text{iff}\ \  & \kstruct,\!\pr,\!\varassign\models\phi \text{ and } \kstruct,\pr,\varassign\models\psi                                          \\
		 & \kstruct,\!\pr,\!\varassign \models \exists^{\lhd n} x\phi                & \text{iff}\ \  & |\{ \delta \mid \kstruct,\!\pr,\!\varassign_{\set{x\mapsto\de}}\models\phi \} | \lhd n \\
		 & \kstruct,\!\pr,\!\varassign \models \standd{\ste}\phi            & \text{iff}\ \  & \kstruct,\!\pr'\!\!,\varassign\models \phi \text{ for some } \pr'\!\!\in\!\sigma(\ste)                                                  \\
		 & \kstruct,\!\pr \models \phi                                      & \text{iff}\ \  & \kstruct,\!\pr,\!\varassign\models\phi \text{ for all } \varassign:\Vars\to\Dom                                                 \\
		 & \kstruct \models \phi                                            & \text{iff}\ \  & \kstruct,\!\pr\models\phi \text{ for all } \pr\in\Precs
	\end{align*}
As usual, $\kstruct$ is a \define{model} for a formula $\phi$ iff \mbox{$\kstruct\models\phi$}. \defend
\end{definition}


\pagebreak

\begin{lemma}
Let $\phi$ be an $\FOformulas$ sentence and $\kstruct=\tuple{\Dom,\Precs,\sigma,\gamma}$ be a model of $\phi$. Then, for any $n \geq |\Precs|$, there exists a model $\kstruct'=\tuple{\Dom,\Precs',\sigma',\gamma'}$ of $\phi$ with $|\Precs'| = n$.     
\end{lemma}

\begin{proof}[Proof Sketch] We just pick one precisification from $\kstruct$ and add as many isomorphic copies of it to $\kstruct$ as required.
\end{proof}



\subsection{Transformations}\label{sec:transformations}

\newcommand{\SmonCtwo}{\ensuremath{\mathbb{S}^\mathrm{mon}_{\Ctwo}}\xspace}

The results obtained in the first part of this paper concern the fragment of all \FOSL formulae that are monodic and \Ctwo{} -- from here on, we will refer to this logical fragment as \emph{monodic standpoint \Ctwo{}}, short \SmonCtwo. 
For technical reasons, we prefer to focus on formulae that additionally are S5, nullary-free, and constant-free; we will call them \emph{frugal}. This section establishes that any \SmonCtwo formula can be efficiently transformed into an equisatisfiable frugal one.

\begin{theorem}\label{theorem:fosl-equisatisfiable-to-S5}
For any \FOSL formula $\phi$, an equisatisfiable S5 \FOSL formula $\mathsf{S5}(\phi)$ can be computed in polynomial time. The transformation preserves $\mathcal{C}^2$-ness and monodicity. 
\end{theorem}

\begin{proof}[Proof Sketch]
Let $\phi$ be a 
\FOSL formula based on a signature $\tuple{\Preds, \Consts, \Stands}$. We show that for any formula $\phi$, the formula $\trans(\phi)$, based on the signature $\tuple{\Preds\cup\Stands, \Consts, \{*\}}$ is equisatisfiable and preserves $\mathcal{C}^2$-ness and monodicity. 
The function $\trans$ replaces every $\standd \psi$ by $\allstandd (e \land \psi)$, introducing one nullary predicate for every standpoint symbol and translating set expressions for standpoints into boolean expressions. The function $\trans$ is defined as follows
\begin{align*}
		\trans(\pred{P}(t_1,\ldots,t_k)) & = \pred{P}(t_1,\ldots,t_k)                                                           \\[-2pt]
		\trans(\lnot\psi)                & =\lnot\trans(\psi)                                                                 \\[-2pt]
		\trans(\psi_1{\,\land\,}\psi_2)  & = \trans(\psi_1){\,\land\,}\trans(\psi_2)                                    \\[-2pt]
		\trans(\forall x\psi)            & = \forall x(\trans(\psi))                                                         \\[-2pt]
		\trans(\standd{\ste}\psi)     & = \standd{*}(\transE(\ste) \land \trans(\psi))
	\end{align*}

    Therein, $\transE$ implements the semantics of standpoint expressions, providing for
each expression \mbox{$\ste\in\StandExps$} a propositional formula $\transE(\ste)$ as follows%
\begin{align*}
	\transE(\sts)                  & = \sts                                       \\[-2pt]
	\transE(\ste_1\cup\ste_2)      & = \transE(\ste_1)\lor\transE(\ste_2)      \\[-2pt]
	\transE(\ste_1\cap\ste_2)      & = \transE(\ste_1)\land\transE(\ste_2)     \\[-2pt]
	\transE(\ste_1\setminus\ste_2) & = \transE(\ste_1)\land\neg\transE(\ste_2)
\end{align*}
 
The proof 
shows equisatisfiability by induction. In addition to this, a routine inspection of the translation ensures that it preserves \Ctwo{}-ness and monodicity, it can be done in polynomial time and its output is of polynomial size. 
This translation is similar in spirit to previous ones, for instance \cite{DBLP:conf/kr/KuruczWZ23}.
\end{proof}

\begin{theorem}\label{theorem:trans-nullary-free}
    For any \FOSL formula $\phi$, one can compute an equisatisfiable nullary-free \FOSL formula $\mathsf{NF}(\phi)$ in polynomial time.
    The transformation preserves $\mathcal{C}^2$-ness, S5-ness, and monodicity. 
\end{theorem}

\begin{proof}[Proof Sketch]
For any nullary predicate $\pred{N}$ occurring in $\phi$, introduce a fresh unary predicate $\pred{P}_\mathtt{N}$ and replace any occurrence of $\pred{N}$ inside $\phi$ by $\forall x.(\pred{P}_\mathtt{N}(x))$. 
\end{proof}

\begin{theorem}\label{theorem:trans-constant-free}
    For any $\mathcal{C}^2$ \FOSL formula $\phi$, one can compute in polynomial time an equisatisfiable constant-free $\mathcal{C}^2$ \FOSL formula $\mathsf{CF}(\phi)$.
    If $\phi$ is S5 and nullary-free, then so is $\mathsf{CF}(\phi)$. 
\end{theorem}

\pagebreak

\begin{proof}[Proof Sketch]
For every constant $\mathtt{a}$ occurring in $\phi$, introduce a unary predicate $\pred{P}_\mathtt{a}$. Let $\phi^\mathrm{consts}$ be the conjunction over all $\exists^{=1}x.\pred{P}_\mathtt{a}(x)\wedge \exists^{=1}x.\allstandb\pred{P}_\mathtt{a}(x)$ for all such $\mathtt{a}$.
Further, obtain $\phi'$ by replacing every atom using constants $\mathtt{a}$, $\mathtt{b}$ 
as follows 
\begin{eqnarray*}
\pred{P}(\mathtt{a}) & \mapsto & \exists x.(\pred{P}_\mathtt{a}(x) \wedge \pred{P}(x)) \\[-1pt]
\pred{P}(\mathtt{a},x) & \mapsto & \exists y.(\pred{P}_\mathtt{a}(y) \wedge \pred{P}(y,x)) \\[-1pt]
\pred{P}(x,\mathtt{a}) & \mapsto & \exists y.(\pred{P}_\mathtt{a}(y) \wedge \pred{P}(x,y)) \\[-1pt]
\pred{P}(\mathtt{a},y) & \mapsto & \exists x.(\pred{P}_\mathtt{a}(x) \wedge \pred{P}(x,y)) \\[-1pt]
\pred{P}(y,\mathtt{a}) & \mapsto & \exists x.(\pred{P}_\mathtt{a}(x) \wedge \pred{P}(y,x)) \\[-1pt]
\pred{P}(\mathtt{a},\mathtt{b}) & \mapsto & \exists x.\exists y.(\pred{P}_\mathtt{a}(x) \wedge \pred{P}_\mathtt{b}(y) \wedge \pred{P}(x,y))\\[-1pt]
x\,\dot{=}\,\mathtt{a}  & \mapsto & \pred{P}_\mathtt{a}(x) \qquad \mbox{(same for }\mathtt{a}\,\dot{=}\,x\mbox{)}\\[-1pt]
y\,\dot{=}\,\mathtt{a}  & \mapsto & \pred{P}_\mathtt{a}(y) \qquad \mbox{(same for }\mathtt{a}\,\dot{=}\,y\mbox{)}\\[-1pt]
\mathtt{a}\,\dot{=}\,\mathtt{b}  & \mapsto & \exists x.(\pred{P}_\mathtt{a}(x) \wedge \pred{P}_\mathtt{b}(x))
\end{eqnarray*}
Then we let $\mathsf{CF}(\phi) = \phi^\mathrm{consts} \wedge \phi'$.
\end{proof}

Thus given an arbitrary \SmonCtwo formula, the consecutive application of the transformations of the above theorems produces an equisatisfiable frugal \SmonCtwo formula. The transformation is polytime and, in particular, the result is of polynomial size with respect to the input.

\section{Satisfiability in Monodic Standpoint \Ctwo{}}\label{sec:sat-in-monodic-c2}

In this section, we study the satisfiability problem of frugal \SmonCtwo and prove \NExpTime-completeness (which carries over to full \SmonCtwo), constituting the central result of the paper.

To get started, we provide an overview of the argument used to establish the result. 
In \Cref{sec:permutational-representatives}, we show that the satisfiability of a frugal $\mathbb{S}^\mathrm{mon}_{\Ctwo}$ formula $\phi$ coincides with the existence of a structure $\kstruct$ having exponentially many precisifications with respect to $\phi$'s size, from which a specific kind of model -- called the \emph{$\Preds_{\rigidp}$-stable permutational closure} of $\kstruct$ -- can be obtained.
In \Cref{sec:stacked-models}, we introduce \emph{stacked interpretations}: plain first-order interpretations that closely reflect the form of standpoint structures for \SmonCtwo. We also define stacked formulae $\phi^m_\mathrm{stack}$, which enforce models to be stacked interpretations corresponding to standpoint structures with $2^m$ precisifications.
With these components in place, we present in \Cref{sec:translating-formulae} an equisatisfiable translation from frugal \SmonCtwo formulae into plain \Ctwo{}, which is polynomial in the size of the input formula.

Throughout the section, we will use a running example to help the reader navigate through the technical details. 

\begin{example}\label{example} Consider the 
     monodic standpoint \Ctwo{} sentence $E$ in \Cref{figure:models}(1), expressing that there is exactly one unequivocally good thing ($E_0$); that everything is either unequivocally good or conceivably the best (somewhere), with no two things being the best simultaneously ($E_1$); and that it is conceivable that everything is good or the best ($E_3$).

\Cref{figure:models}(2) shows a model of $E$. Notably, in models of $E$ with infinite domains -- such as the one in \cref{figure:models}(2) -- there must also be infinitely many precisifications. This is because only one element satisfies $\pred{Good}$ everywhere, while every other element must be the $\pred{Best}$ in some precisification, with at most one such element per precisification. \defend

\begin{figure*}[h]
\centerline{\includegraphics[width=\textwidth]{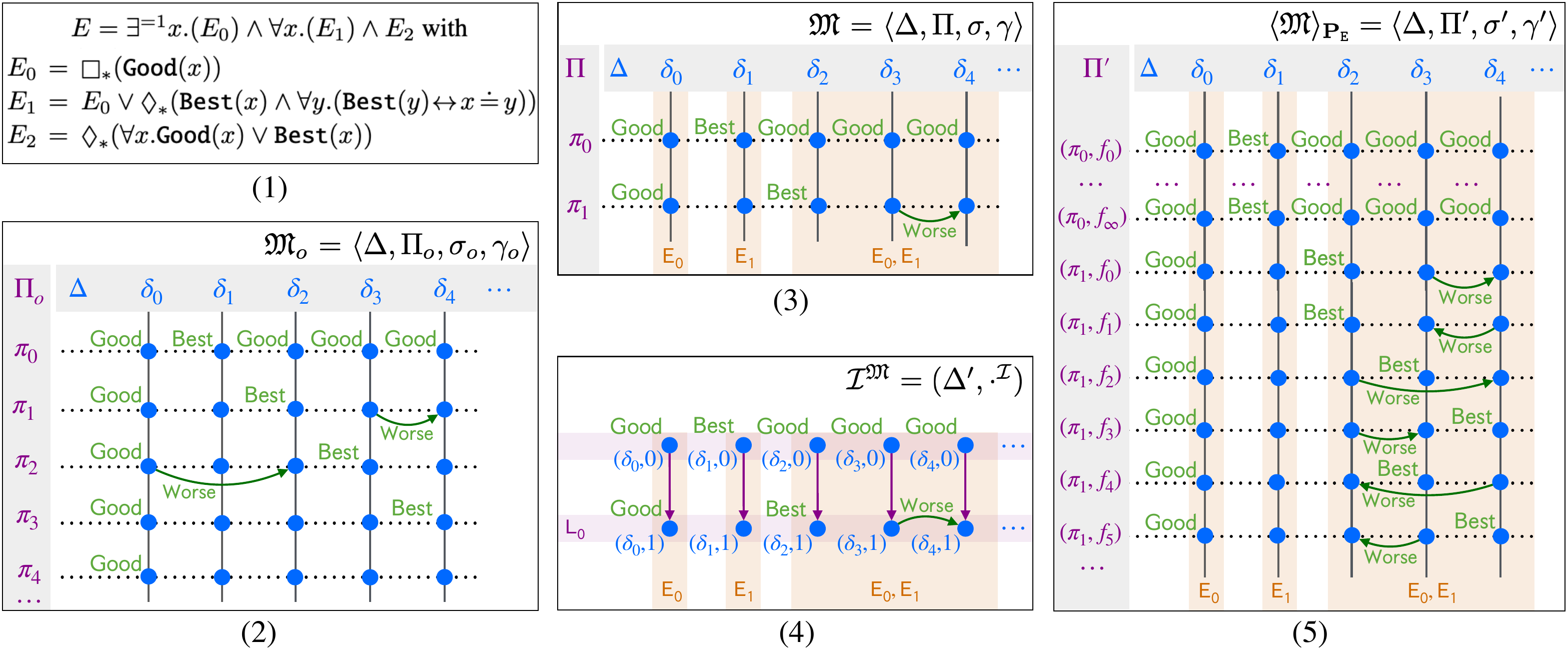}}
\caption{(1) The formula $E$ from \Cref{example}, and illustrations of (2) a model $\kstruct_{o}$ of $E$, (3) an interpretation $\kstruct$ with the signature $\tuple{\Preds \uplus \Preds_{\rigidp},\emptyset,\{*\}}$, (4) the stacked interpretation of $\kstruct$, $\mathcal{I}^{\kstruct}$ and (5) $\modelPR$, the $\Preds_{\rigidp}$-stable permutational closure of $\kstruct$. In the graphics, all points within a coloured area labelled with a unary predicate are in the interpretation of that predicate.\\[-3ex]}\label{figure:models}
\end{figure*}
\end{example}

\pagebreak

\subsection{Permutational Representatives}\label{sec:permutational-representatives}

Next, we show that for any satisfiable frugal $\mathbb{S}^\mathrm{mon}_{\Ctwo}$ formula~$\phi$, there is a structure $\kstruct$ with only exponentially many precisifications in $|\phi|$ from which a model of $\phi$ of a specific shape can be created (while $\kstruct$ may not be a model itself). 

\begin{definition}\label[definition]{def:perm-closure}
Let $\kstruct=\tuple{\Dom,\Precs,\sigma,\gamma}$ be a standpoint structure for the signature $\tuple{\Preds \uplus \Preds_{\rigidp},\emptyset,\{*\}}$,
where $\Preds$ contains only unary and binary predicates, and $\Preds_{\rigidp}=\{\rigidp_0,\ldots,\rigidp_\ell\}$ is a set of special rigid unary predicates.
Let $\mathbb{P}_{\rigidp}$ denote the set of all permutations (i.e., bijective functions) $f : \Delta \to \Delta$ which preserve (non)membership in every~$\rigidp_i$, that is, for every  
$i \in \{0,\ldots,\ell\}$ and $\delta \in \Delta$, we require $\delta \in \rigidp_i^{\kstruct} \Leftrightarrow f(\delta) \in \rigidp_i^{\kstruct}$.  

Then, the \emph{$\Preds_{\rigidp}$-stable permutational closure} of $\kstruct$, denoted $\modelPR$ is the standpoint structure $\tuple{\Dom,\Precs',\sigma',\gamma'}$ defined by
\begin{itemize}
\item $\Precs' = \Precs \times \mathbb{P}_{\rigidp}$,
\item $\sigma'(\sts) = \sigma(\sts) \times \mathbb{P}_{\rigidp}$,
\item $\pred{P}^{\gamma'((\pi,f))} = \{ f(\delta) \mid \delta \,{\in}\, \pred{P}^{\gamma(\pi)}\}$ for unary predicates $\pred{P} \,{\in}\, \Preds$ 
\item $\pred{P}^{\gamma'((\pi,f))} = \{ (f(\delta_1),f(\delta_2)) \mid (\delta_1,\delta_2) \,{\in}\, \pred{P}^{\gamma(\pi)}\}$ for binary predicates $\pred{P} \,{\in}\, \Preds$ \defend
\end{itemize} 

\end{definition}

As we can see, the structure $\kstruct$ contains a set of special rigid unary predicates $\Preds_{\rigidp}$. These predicates induce ``$\rigidp$-types'', corresponding to the subsets $T \subseteq \Preds_{\rigidp}$, so a domain element is said to have the $\rigidp$-type $T$ if it belongs to the interpretation of each $\rigidp_i$ in $T$ and to none outside it. We say $T$ is realized (in $\kstruct$) if at least one domain element has it.

The {$\Preds_{\rigidp}$-stable permutational closure} of $\kstruct$ produces a much larger structure that contains, for each initial precisification in $\Precs$, the set of precisifications with all possible permutations between domain elements belonging to the same $\rigidp$-type. Locally, all permuted versions of any $\pr\in\Precs$ in the closure are isomorphic to each other, they just have their elements ``swapped around'', preserving the internal structure. This intuition is materialised in the lemma below.


\begin{lemma}\label[lemma]{lemma:permutations-model-the-same}
Let $\phi$ be a frugal \SmonCtwo formula and let $\modelPR$ the \emph{$\Preds_{\rigidp}$-stable permutational closure} of a standpoint structure $\kstruct$. 
Let $(\pr,f)$ and $(\pr,f')$ be precisifications of $\modelPR$ and
$\varassign'=f'\circ f^{-1}\circ\varassign$. Then, 
$$\modelPR,\!(\pr,f),\!\varassign \models \phi \Longleftrightarrow \modelPR,\!(\pr,f'),\!\varassign' \models \phi.$$
\end{lemma}


At the global level, $\rigidp$-types of domain elements, which are preserved under the permutations (by construction), will be utilized to witness the elements' ``membership'' in diamond-preceded subformulae in the following sense: we call some $\delta \in \Delta$ a \emph{member} of some formula $\allstandd\psi$ with one free variable $z$ if $\kstruct,\pr,z \mapsto \delta \models \allstandd\psi$ for some/all $\pr \in \Pi$ (note that by the semantics, the choice of $\pr$ is irrelevant in this case). We will denote the set of members of $\allstandd\psi$ in $\kstruct$ by $(\allstandd\psi)^{\kstruct}$.



Let us now investigate what conditions must $\kstruct$ meet so that $\modelPR$ is a model of $\phi$. First, for $\modelPR$ to witness membership to the $\allstandd\psi$ formulae, the number of $\rigidp$-types must be at least as large as the number of types induced by the monodic modal formulae -- these we refer to as $\Diamond$-types. The two sets of types will be aligned: each $\rigidp$-type will be entirely contained within a corresponding $\Diamond$-type. 
Moreover, for each formula $\allstandd\psi$ in a given $\Diamond$-type, $\kstruct$ must include at least one precisification in which some element of that type satisfies $\psi$. The {$\Preds_{\rigidp}$-stable permutational closure} then ensures that every other element of the same $\rigidp$-type also satisfies $\psi$ in some permutation of that precisification.


\begin{theorem}\label{theorem:sat-in-permutational-closure}
Let $\phi$ be a satisfiable frugal \SmonCtwo formula over the signature $\tuple{\Preds,\emptyset,\{*\}}$.
Let $Dia_{\phi}$ denote the diamond subformulae of $\phi$ and $\mathit{FreeDia}_{\phi}$ the diamond subformulae with one free variable.
Then there is a standpoint structure $\kstruct=\tuple{\Dom,\Precs,\sigma,\gamma}$ over $\tuple{\Preds \uplus \Preds_{\rigidp},\emptyset,\{*\}}$ with
\begin{itemize}
\item $|\Preds_{\rigidp}| = \ell = |\mathit{FreeDia}_{\phi}|$
\item $|\Precs| \leq |Dia_{\phi}|\cdot 2^{|Dia_{\phi}|} $
\end{itemize}
such that $\modelPR$ is a model of $\phi$. 
\end{theorem}

\begin{proof}[Proof Sketch]
To prove \Cref{theorem:sat-in-permutational-closure}, we start from an arbitrary model $\kstruct'$ of $\phi$ and let $\mathit{FreeDia}_{\phi}=\{\allstandd\phi_1, \ldots \allstandd\phi_\ell\}$ be the set of diamond subformulae of $\phi$ with one free variable.
We enrich $\kstruct'$ by $\Preds_{\rigidp}$, setting the extension of $\rigidp_i$ to $(\allstandd\phi_i)^{\kstruct'}$ for every $i \in \{1,\ldots,\ell\}$.
Then we create a new structure $\kstruct$ by selecting at most exponentially many precisifications from the enriched $\kstruct'$ and removing the rest. Specifically,
\begin{itemize}
    \item select an arbitrary $\pr$ in case there are no diamond subformulae of $\phi$ at all. Otherwise,
    \item for each $\allstandd\psi$ with no free variables that is satisfied in $\kstruct'$, select some $\pr$ with $\kstruct',\pr\models\psi$ and
    \item for each realised $\rigidp$-type $T \subseteq \Preds_{\rigidp}$, pick some $\delta$ that has $T$, and select, for every $\rigidp_i \in T$ one $\pr$ with $\kstruct',\pr,z\mapsto \delta \models \phi_i$
\end{itemize}
The first point ensures that $\Precs$ is nonempty. The second adds witnesses for sentential modal formulae. The third provides witnesses of all monodic modal formulae from $\mathit{FreeDia}_{\phi}$. The construction ensures that at least one domain element witnesses each $\allstandd\phi_i$ formula of each $\Diamond$-type. Once these seed witnesses are in $\kstruct$, the rest of the elements belonging to that type in $\kstruct'$ will be witnessed by a permutation in $\modelPR$. One can then show by induction on the structure of $\phi$ that $\modelPR$ is a model iff $\kstruct'$ is a model. 
\end{proof}
\vspace{-1ex}
\begin{example} Revisiting \Cref{example}, note that $E$ contains two monodic modal subformulae, $E_0$ and $E_1$. From the model of $E$ shown in \cref{figure:models}(2), we can extract a structure~$\kstruct$ with\linebreak \mbox{$\Preds_\rigidp = \{\rigidp_0, \rigidp_1\}$} (depicted in \cref{figure:models}(3)), such that the corresponding model $\modelPR$ (shown in \cref{figure:models}(5)) also satisfies $E$.

In constructing $\kstruct$, we proceed as follows:
\begin{itemize}
    \item We use $\pr_0$ to witness the sentential modal subformula $E_2$.
    \item The type $\{\}$ is not realised.
    \item For type $\{\pred{E}_0\}$, we use $\pr_0$ as witnes since $\delta_0 \in \pred{Good}^{\gamma{o}(\pr_0)}$.
    \item For type $\{\pred{E}_1\}$ we use $\pr_0$ as witnes since $\delta_1\in\pred{Best}^{\gamma_{o}(\pr_0)}$.
    \item For type $\{\pred{E}_0,\pred{E}_1\}$ we use $\pr_0$ and $\pr_1$ as witnesses since $\delta_2\in\pred{Good}^{\gamma_{o}(\pr_0)}$ and $\delta_2\in\pred{Best}^{\gamma_{o}(\pr_1)}$. \defend
\end{itemize}
Notice that the permutations of $\delta_2$ ensure that the membership to the formulae $E_0$ and $E_1$ in $\kstruct$ carries on to $\modelPR$.
\end{example}
\vspace{-1ex}

\subsection{Stacked Interpretations}\label{sec:stacked-models}

We now define a specific kind of \Ctwo{} interpretation obtained from a given standpoint structure $\kstruct$ with $2^m$ precisifications, called the \emph{stacked interpretation} of $\kstruct$ and denoted $\mathcal{I}^{\kstruct}$. This structure is designed to closely mirror the shape of $\kstruct$.

\begin{definition}\label[definition]{def:stacked-model}
Let $\kstruct=\tuple{\Dom,\Precs,\sigma,\gamma}$ with $|\Precs| {=} 2^m$ be a standpoint structure for the signature $\tuple{\Preds,\emptyset,\{*\}}$ where $\Preds$ contains only unary and binary predicates. Assume $\Precs$ is linearly ordered with elements named $\pi_0,\pi_1,\ldots,\pi_{2^m-1}$.

The \emph{stacked interpretation} of $\kstruct$ is the FO-interpretation $\mathcal{I}^{\kstruct} = (\Delta', \cdot^\mathcal{I})$ with signature $\tuple{\Preds \uplus \{\pred{F},\pred{L}_0, \ldots, \pred{L}_{m-1}\},\emptyset}$, where $\pred{F}$ is a fresh binary predicate and $\pred{L}_0, \ldots, \pred{L}_{m-1}$ are fresh unary predicates, such that
\begin{enumerate}[ label={(S\arabic*)}, ref={(S\arabic*)}]
\item $\Delta' = \Delta \times \{0,\ldots,2^m-1\}$\label{def-enum:stacked-domain}
\item $\pred{L}_j^\mathcal{I} = \{ (\delta,i) \mid \mbox{the $j^\mathrm{th}$ bit of $i$ in binary encoding is $1$}\}$   \label{def-enum:stacked-Ls}
\item $\pred{F}^\mathcal{I} = \{ ((\delta,i),(\delta,{i+1})) \mid \delta \in \Delta,\ 0\leq i < 2^m-1 \}$  \label{def-enum:stacked-F}
\item $\pred{P}^\Inter = \bigcup_{0 \leq i < 2^m} \pred{P}^{\gamma(\pi_i)} \times \{i\}$ \ for all unary $\pred{P} \in \Preds$,  \label{def-enum:stacked-unary-preds}
\item $\pred{P}^\Inter =  
\{((\delta_1,i),(\delta_2,i)) \mid 0 \,{\leq}\, i \,{<}\, 2^m, (\delta_1,\delta_2) \,{\in}\, \pred{P}^{\gamma(\pi_i)} \}
$ for all binary $\pred{P} \in \Preds$. \defend \label{def-enum:stacked-binary-preds}
\end{enumerate}
\end{definition}
\newpage

Our approach constructs a stacked domain by creating one copy of the original domain $\Delta$ for each precisification in $\kstruct$, so that each new element $(\delta, \pr)$ mimics $\delta$ at precisification $\pr$. A set of new unary predicates $\pred{L}_0, \ldots, \pred{L}_{m-1}$ encodes the index of the associated precisification (each less than $2^m$). Additionally, a new binary predicate $\pred{F}$ links each element $(\delta, \pr_i)$ to its successor $(\delta, \pr_{i+1})$. Thus, for every original element $\delta$, the stacked interpretation forms an $\pred{F}$-chain tracking $\delta$ across all precisifications in $\kstruct$. In \Cref{figure:models}, (3) depicts the stacked interpretation of (2) with the $\pred{F}$-chains in purple.

\begin{definition}
For a given $m \in \mathbb{N}$, we define $\phi^m_\mathrm{stack}$ as the conjunction of the following formulae
\begin{enumerate}[ label={(F\arabic*)}, ref={(F\arabic*)}]
\item $\forall x.(\bigvee_{0\leq j< m} \neg \pred{L}_j(x)) \to \exists^{=1}y.\pred{F}(x,y)$ \label{def-enum:everything-but-last-has-next}
\item $\forall x.(\bigwedge_{0\leq j< m} \pred{L}_j(x)) \to \exists^{=0}y.\pred{F}(x,y)$ \label{def-enum:last-has-no-next}
\item $\forall x.(\bigvee_{0\leq j< m} \pred{L}_j(x)) \to \exists^{=1}y.\pred{F}(y,x)$ \label{def-enum:everything-but-first-has-previous}
\item $\forall x.(\bigwedge_{0\leq j< m} \neg \pred{L}_j(x)) \to \exists^{=0}y.\pred{F}(y,x)$ \label{def-enum:first-has-no-previous}

\item 
\label{def-enum:level-counter}
	$\displaystyle\forall x y.\pred{F}(x,y) \to \bigwedge_{\mathclap{0\leq j < m}} \Big(\! \big(\pred{L}_{j}(x) \,{\leftrightarrow}\, \pred{L}_{j}(y)\big) \leftrightarrow \bigvee_{\mathclap{0\leq j'<j}} \neg \pred{L}_{j'}(x) \!\Big)   $ 

\item $\forall x y. \pred{P}(x,y) \!\to\! \displaystyle{\bigwedge_{\mathclap{0\leq j< m}}}\,\pred{L}_j(x)\,{\leftrightarrow}\,\pred{L}_j(y)$ for all binary $\pred{P} \,{\in}\, \Preds$\!. \label{def-enum:binary-at-same-level}
\end{enumerate} 
\vspace{-2ex}
~\defend
\end{definition}

The stacked formula of size $m$, denoted $\phi^m_\mathrm{stack}$, is used to enforce that models are stacked models. Clause~\ref{def-enum:everything-but-last-has-next} enforces that all elements except those with the highest index (as determined by the $\pred{L}$ predicates) have exactly one $\pred{F}$-successor. Conversely, \ref{def-enum:last-has-no-next} ensures that elements with the highest index have none. Clauses~\ref{def-enum:everything-but-first-has-previous} and~\ref{def-enum:first-has-no-previous} impose analogous constraints on $\pred{F}$-predecessors. Clause~\ref{def-enum:level-counter} encodes that any two $\pred{F}$-connected elements have consecutive indexes, via a binary level-counter using the $\pred{L}$ predicates. Lastly, Clause~\ref{def-enum:binary-at-same-level} enforces that all binary predicates (except $\pred{F}$) relate only elements with matching indices.

\begin{lemma}\label[lemma]{lemma:stacked-model-satisfies-stack-formula}
    Any stacked interpretation $\mathcal{I}^{\kstruct}$ satisfies $\phi^m_\mathrm{stack}$.
\end{lemma}

\begin{proof}[Proof sketch] We verify that each clause of $\phi^m_\mathrm{stack}$ is satisfied by the stacked interpretation $\mathcal{I}^{\kstruct}$ as defined. In particular, the structure of the domain, and the interpretation of the predicates $\pred{F}$, $\pred{L}_0,\dots,\pred{L}_{m-1}$, and $\pred{P}\in\Preds$ ensure that all required properties hold. 
\end{proof}

\begin{theorem}\label{lemma:model-satisfies-stack-formula-if-isomorphic-to-stacked-model}
A first-order interpretation $\mathcal{I}$ over the signature $\tuple{\Preds \uplus \{\pred{F},\pred{L}_0, \ldots, \pred{L}_{m-1}\},\emptyset}$ satisfies $\phi^m_\mathrm{stack}$ if and only if it is isomorphic to a stacked interpretation $\mathcal{I}^{\kstruct}$ of some standpoint structure $\kstruct$ over signature $\tuple{\Preds,\emptyset,\{*\}}$ with $2^m$ precisifications.
\end{theorem}

\begin{proof}[Proof Sketch]
If $\mathcal{I}$ is isomorphic to a stacked model $\mathcal{I}^{\kstruct}$ then it satisfies $\phi^m_{\smash{\mathrm{stack}}}$ by \Cref{lemma:stacked-model-satisfies-stack-formula}.
It remains to prove the other direction, i.e., for any $\mathcal{I}$ that satisfies $\phi^m_{\smash{\mathrm{stack}}}$, there exists a standpoint structure $\kstruct$ for which $\mathcal{I}^{\kstruct}$ is isomorphic to $\mathcal{I}$. We show how to construct $\kstruct=\tuple{\Dom,\Precs,\sigma,\gamma}$ given $\mathcal{I}$ with domain $\Delta'$:
For any $\delta' \in \Delta'$ we let $\mathsf{level}(\delta')$ denote the unique number $i < 2^m$ that satisfies, for every $j <m$, that the $(j+1)$th bit in the binary encoding of $i$ is $1$ if and only if \smash{$\delta' \in \pred{L}_j^\mathcal{I}$}.
Moreover, we let $\approx$ be the smallest equivalence relation containing $\pred{F}^\mathcal{I}$ and let $\Delta$ consist of the $\approx$-equivalence classes of $\Delta'$
. As in Definition~\ref{def:stacked-model},
let $\Precs=\{\pi_0,\ldots,\pi_{2^m-1}\}$. Obviously, $\sigma=\{*\mapsto \Precs\}$.
Finally, we set $\pred{P}^{\gamma(\pi_i)}$ to \pagebreak
\begin{itemize}
\item 
$\{[\delta']_{\approx} \mid \delta'{\,\in\,} \pred{P}^\mathcal{I}, \mathsf{level}(\delta'){\,=\,}i \}$ for unary $\pred{P} {\,\in\,}\Preds$,
\item
$\{([\delta'_1]_{\approx},[\delta'_2]_{\approx}) \mid (\delta'_1,\delta'_2){\,\in\,} \pred{P}^\mathcal{I}, \mathsf{level}(\delta_1'){\,=\,}\mathsf{level}(\delta_2'){\,=\,}i \}$ for binary $\pred{P} {\,\in\,}\Preds$.
\end{itemize}
Then, the bijection $\stacked:\Delta'\to\Delta \times \{0,\ldots,2^m-1\}$
defined via $\stacked(\delta') = ([\delta']_\approx,\mathsf{level}(\delta'))$ can be shown to constitute an isomorphism from $\mathcal{I}$ to $\mathcal{I}^{\kstruct}$.
%
%
%
%
\end{proof}

\subsection{Translating Formulae}\label{sec:translating-formulae}

So far, we have shown that the satisfiability of a frugal \SmonCtwo formula $\phi$ coincides with the existence of a structure $\kstruct$ of  size exponential in $|\phi|$ from which a model can be extracted. Furthermore, we demonstrated that such structures can be characterized in plain \Ctwo{} through their corresponding stacked interpretations. In this subsection, we leverage these results to define a translation from  frugal \SmonCtwo into plain \Ctwo{}, such that a frugal \SmonCtwo formula $\phi$ is satisfiable if and only if its translation into \Ctwo{} is satisfiable. Together with the translation from \SmonCtwo to frugal \SmonCtwo, this entails the upper complexity bound for all of \SmonCtwo. 

\begin{definition}\label{def:translation}
Given some $m\in \mathbb{N}$, we define the function $\mathsf{Trans}_m$ that maps frugal \SmonCtwo sentences $\phi$ over the signature $\tuple{\Preds \uplus \Preds_{\rigidp},\emptyset,\{*\}}$ with $\Preds_{\rigidp}=\{\rigidp_0,\ldots\rigidp_\ell\}$ into \Ctwo{} as follows: 
$\mathsf{Trans}_m(\phi)$ is the sentence $\forall x.\forall y.(x\,\dot{=}\,y \to \mathsf{tr}(\phi))$, where the function $\mathsf{tr}$ is recursively defined via
\newcommand{\mapsnarrow}{\!\!\!\mapsto\!\!\!}
\begin{eqnarray*}
\psi & {\mapsnarrow} & \psi \mbox{\ \ \ \ if $\psi$ is of the form $\pred{P}(z)$, $\pred{P}(z,z')$ or $z\,\dot{=}\,z'$ } \\
 \neg \psi & {\mapsnarrow} & \neg(\mathsf{tr}(\psi))\\
 \psi \wedge \psi' & {\mapsnarrow} & \mathsf{tr}(\psi) \wedge \mathsf{tr}(\psi')\\
 \exists^{\lhd n} z.\psi & {\mapsnarrow} & \exists^{\lhd n} z. (\phi^=_\mathtt{L}(x,y) \wedge  \mathsf{tr}(\psi))\\
  {\Diamond}_* \psi & {\mapsnarrow} & \forall z_\mathrm{nf}.x\,\dot{=}\,y \to  \exists z_{{\mathrm{mf}}}.\phi^=_\rigidp(x,y) \wedge \mathsf{tr}(\psi)\qquad
\end{eqnarray*}	
where $z,z' \in \{x,y\}$ and
\begin{itemize}
\item $z_\mathrm{nf}$ is a variable from $\{x,y\}$ that is not free in $\psi$ and $\{z_\mathrm{mf}\} = \{x,y\} \setminus \{z_\mathrm{nf}\}$ ,
\item 	
$\phi^=_\mathtt{L}(x,y)$ abbreviates $\bigwedge_{0\leq j< m} \pred{L}_j(x) \leftrightarrow \pred{L}_j(y)$, and
\item 
$\phi^=_\rigidp(x,y)$ abbreviates $\bigwedge_{0\leq i\leq \ell} \rigidp_i(x) \leftrightarrow \rigidp_i(y)$.\defend
\end{itemize}
\end{definition}

The key components of the translation are the handling of counting quantification and modal operators. The translation of an counting existential quantification employs the formula $\phi^=_\mathtt{L}(x,y)$ to ensure that quantification ranges only over elements belonging to the current layer of the stacked interpretation -- namely, those whose counterparts correspond to the domain elements at the current precisification. In contrast, the translation of modal subformulae of the form $\allstandd\psi$ makes use of $\phi^=_\rigidp(x,y)$ to ensure that quantification ranges over the elements belonging to the current $\rigidp$-type. Recall that $\mathcal{I}^{\kstruct}$ is constructed to mirror the structure $\kstruct$, from which in turn we obtain the model $\modelPR$. Consequently, if any element of the same $\rigidp$-type satisfies $\psi$, then there exist some permutation within $\modelPR$ that satisfies $\psi$ and thus the formula $\allstandd\psi$ is satisfied. Notice that when $\allstandd\psi$ is sentential, the variable assignment does not make a difference. The following Lemma formally establishes the discussed correspondence. 

\begin{lemma}\label[lemma]{lemma:permutational-closure-iff-stacked-model}
Let $\phi$ be a frugal \SmonCtwo sentence over the signature $\tuple{\Preds,\emptyset,\{*\}}$. Let $\kstruct=\tuple{\Dom,\Precs,\sigma,\gamma}$ be a standpoint structure for the signature $\tuple{\Preds \uplus \Preds_{\rigidp},\emptyset,\{*\}}$, with all predicates from $\Preds_{\rigidp}=\{\rigidp_0,\ldots,\rigidp_\ell\}$ rigid, and $|\Precs| = 2^m$
. Then,
$$\modelPR \models \phi \Longleftrightarrow \mathcal{I}_\kstruct \models \mathsf{Trans}_m(\phi).$$ 
\end{lemma}

\begin{proof}[Proof Sketch]
Toward the result, we first prove the claim that, for $\pr_i\in\Precs$, $f\in\mathbb{P}_{\rigidp}$ and a variable assignment $\varassign$, we have $\modelPR
,(\pr_i,f),\varassign\models\phi$ iff $\mathcal{I}_\kstruct,\varassign' \models \mathsf{tr}(\phi)$ where $\varassign'(z)=(f(\varassign(z)),i)$ for $z\in\{x,y\}$. Due to \Cref{lemma:permutations-model-the-same}, it suffices to show that $\modelPR
,(\pr_i,f_{id}),\varassign\models\phi$ iff $\mathcal{I}_\kstruct,\varassign' \models \mathsf{tr}(\phi)$ where $\varassign'(z)=(\varassign(z),i)$
; this follows by structural induction on $\phi$.

Toward the statement of the Lemma, assume $\modelPR \models \phi$, thus $\modelPR
,(\pr_i,f),\varassign\models\phi$ holds for all $\pr_i\in\Precs$, $f\in\mathbb{P}_{\rigidp}$ and assignments $\varassign$. Then, $\mathcal{I}_\kstruct,\varassign' \models \mathsf{tr}(\phi)$ where $\varassign'(z)=(\varassign(z),i)$, thus $\mathcal{I}_\kstruct,\varassign' \models \mathsf{tr}(\phi)$ for all $\varassign'$ where $x$ and $y$ have equal index, thus $\mathcal{I}_\kstruct \models \mathsf{Trans}_m(\phi)$. The converse direction proceeds similarly.
%
\end{proof}

\begin{example}
Revisiting \Cref{example}, we compute $\mathsf{Trans}_2(E)$.\\ With some simplifications, we obtain the following:

\medskip
\newcommand{\nar}{=}
\noindent$\begin{array}{@{}r@{\ }l@{\ }l@{}}
\mathsf{Trans}_2(E) & \nar &\forall x.\forall y.x\,\dot{=}\,y \to \mathsf{tr}(E) \\[3pt]
\mathsf{tr}(E)    & \nar & \exists^{=1} x. ( \phi^=_\mathtt{L}(x,y) \wedge \mathsf{tr}(E_0)) \\[0pt]
                  &   & \ \wedge\ \forall x.\left( \phi^=_\mathtt{L}(x,y) \to \mathsf{tr}(E_1)\right) \wedge \mathsf{tr}(E_2)\\[3pt] 
\mathsf{tr}(E_0)  & \nar & \exists y.\, x \dot{=} y \wedge \forall x.\, \phi^=_\rigidp(x,y) \to \pred{Good}(x)\\[3pt]
\mathsf{tr}(E_1)  & \nar &  \mathsf{tr}(E_0) \vee \forall y. x \dot{=} y \to (\exists x.\, \phi^=_\rigidp(x,\!y)\, {\wedge}\\
  &  &
\pred{Best}(x) {\wedge}  \forall y.\, \phi^=_\mathtt{L}(x,y)\! \to \!\left( \pred{Best}(y) \!\leftrightarrow \!x \dot{=} y 
\right))\\[3pt]
\mathsf{tr}(E_2)  & \nar &  \forall x.\,  x \dot{=} y  \to \exists y.\big(\phi^=_\rigidp(x,\!y) \wedge \\[-0pt]
                  &   & \ \forall x. \phi^=_\mathtt{L}(x,y) \to  \left( \pred{Good}(x) \vee \pred{Best}(x) \right)\big)
\end{array}$\\

%
%
%
\noindent One may verify that the structure in \Cref{figure:models}(4) indeed satisfies $\mathsf{Trans}_2(E)$. Roughly, there exists $x$, e.g., $(\delta_0,0)$, s.t. all elements of its $\pred{E}$-Type, i.e. $(\delta_0,0)$ and $(\delta_0,1)$, are $\pred{Good}$ ($E_0$), and for all elements, either they satisfy  $E_0$ (like $(\delta_0,0)$ and $(\delta_0,1)$) or there is some element of their $\pred{E}$-Type (e.g., $(\delta_1,0)$ for $\{\pred{E}_1\}$ and $(\delta_2,1)$ for $\{\pred{E}_0,\pred{E}_1\}$) which is the only $\pred{Best}$ element on their layer ($E_1$). Finally there is some element (e.g., $(\delta_0,0)$) such that all elements on its layer are $\pred{Good}$ ($E_3$).\defend
\end{example}

The last ingredient for our satisfiability translation is to ensure that the predicates in $\Preds_{\rigidp}$ are indeed rigid. 

\begin{definition}
We let $\phi^\ell_{rig\rigidp}$ denote the \Ctwo{} sentence $\forall x.\forall y.\mathtt{F}(x,y) \to \bigwedge_{0\leq i\leq \ell} \rigidp_i(x)\leftrightarrow \rigidp_i(y)$.
\end{definition}

\begin{lemma}\label[lemma]{lemma:permutational-closure-iff-rigid}
Let $\kstruct=\tuple{\Dom,\Precs,\sigma,\gamma}$ be a \foss for the signature $\tuple{\Preds \uplus \Preds_{\rigidp},\emptyset,\{*\}}$. Then, all predicates from $\Preds_{\rigidp}=\{\rigidp_0,\ldots,\rigidp_\ell\}$ are rigid iff
$ \mathcal{I}_\kstruct \models \phi^\ell_{rig\rigidp}$. 
\end{lemma}

\begin{theorem}\label{thm:satcorrespondence}
Let $\phi$ be an arbitrary frugal \SmonCtwo sentence with $|\mathit{FreeDia}_{\phi}|=\ell$ and $\lceil |Dia_{\phi}| + \log_2 (|Dia_{\phi}|) \rceil=m$.
Then $\phi$ is satisfiable iff $\phi^m_\mathrm{stack} \wedge \phi^\ell_{rig\rigidp} \wedge \mathsf{Trans}(\phi)$ is satisfiable.   
\end{theorem}

\begin{proof}
Assume that $\phi$ is satisfiable. Then by \Cref{theorem:sat-in-permutational-closure} there is a standpoint structure $\kstruct$ over the signature $\tuple{\Preds \uplus \Preds_{\rigidp},\emptyset,\emptyset}$ with
$|\Preds_{\rigidp}| = \ell = |\mathit{FreeDia}_{\phi}|$
and $|\Precs| \leq |Dia_{\phi}|\cdot 2^{|Dia_{\phi}|}$,
such that $\modelPR$ is a model of $\phi$. Then, from \Cref{lemma:permutational-closure-iff-stacked-model}, we have that $\mathcal{I}_\kstruct \models \mathsf{Trans}(\phi)$. Moreover, from \Cref{lemma:stacked-model-satisfies-stack-formula}, $\mathcal{I}_\kstruct \models \phi^m_\mathrm{stack}$. Finally, by \Cref{def:perm-closure}, all predicates from $\Preds_{\rigidp}$ are rigid and thus by \Cref{lemma:permutational-closure-iff-rigid}
$\mathcal{I}_\kstruct \models \phi^\ell_{rig\rigidp}$.

For the other direction, assume that there is a model $\mathcal{I}$ over the signature $\tuple{\{\pred{F},\pred{L}_0, \ldots, \pred{L}_{m-1}\} \uplus \Preds_{\rigidp}\uplus\Preds,\emptyset}$ such that $\mathcal{I}\models\phi^m_\mathrm{stack} \wedge \phi^\ell_{rig\rigidp} \wedge \mathsf{Trans}(\phi)$. Then, by \Cref{lemma:model-satisfies-stack-formula-if-isomorphic-to-stacked-model}, $\mathcal{I}$ is isomorphic to a stacked interpretation $\mathcal{I}^{\kstruct}$ of some  standpoint structure $\kstruct$ over $\tuple{\Preds \uplus \Preds_{\rigidp},\emptyset,\emptyset}$  with $2^m$ precisifications. Moreover, since $\mathcal{I}\models\phi^\ell_{rig\rigidp}$ then by \Cref{lemma:permutational-closure-iff-rigid} the predicates in $\Preds_{\rigidp}$ are rigid. And since $\mathcal{I}^{\kstruct}\models\mathsf{Trans}(\phi)$, then by \Cref{lemma:permutational-closure-iff-stacked-model}, we have that $\modelPR\models\phi$ as desired.
\end{proof}

Therefore (and taking into account \Cref{sec:transformations}), there is a polytime equisatisfiable translation from \SmonCtwo to plain \Ctwo{}. On the other hand, every plain \Ctwo{} formula is \SmonCtwo, thus the below corollary follows from the known \NExpTime completeness of plain \Ctwo{} \cite{Pratt-Hartmann05}.

\begin{corollary}
Satisfiability in monodic standpoint \Ctwo{} is  \NExpTime-complete. 
\end{corollary}

Our equisatisfiable ``standpoint removal'' technique turns out to be robust under some variations. Let us call a monodic standpoint \Ctwo{} formula $\phi$ \emph{finitely satisfiable} if it has a model $\kstruct=\tuple{\Dom,\Precs,\sigma,\gamma}$ where $\Delta$ is finite. It is easy to see that all equisatisfiable transformations in \Cref{{sec:transformations}} are also ``equi-finitely-satisfiable'', because the underlying model transformations do not alter the domain whatsoever; the same holds for the argument behind \Cref{theorem:sat-in-permutational-closure}. Last not least, the domain $\Delta'$ of the stacked interpretation $\mathcal{I}^{\kstruct} = (\Delta', \cdot^\mathcal{I})$ corresponding to a structure $\kstruct=\tuple{\Dom,\Precs,\sigma,\gamma}$ is finite whenever $\Delta$ is (by the construction of \Cref{def:stacked-model}, we get $|\Delta'| = |\Delta|\cdot 2^m$). Thus the correspondence established in \Cref{thm:satcorrespondence} also holds for finite satisfiability. On the other hand, finite satisfiability of plain \Ctwo{} is also known to be \NExpTime-complete \cite{Pratt-Hartmann05}, thus we obtain the following result.

\begin{corollary}
Finite satisfiability in monodic standpoint \Ctwo{} is \NExpTime-complete. 
\end{corollary}

Last not least, more recently \citeauthor{ultimatelyperiodic} (\citeyear{ultimatelyperiodic}) considered two-variable FO with a more expressive version of counting quantifiers, denoted $\exists^S$, where $S$ is any semilinear subset of $\mathbb{N} \cup \{\infty\}$. For example, by means of such quantifiers one can express quantities like ``evenly many $x$'' or also ``infinitely many $x$'', which go beyond what can be stated by the counting quantifiers of \Ctwo. Satisfiability of the ensuing logic, denoted $\mathrm{FO}^{\smash{2}}_\mathrm{Pres}$ was established to be decidable in \NTwoExpTime and \NExpTime hard. We note that our definitions, constructions, and arguments seamlessly extend from \Ctwo{} to this logic, leading to the subsequent corollary.

\begin{corollary}
Satisfiability and finite satisfiability in monodic standpoint $\mathrm{FO}^{\smash{2}}_\mathrm{Pres}$ is in \NTwoExpTime and hard for \NExpTime. 
\end{corollary}

We believe that -- beyond their applicability to ontology reasoning as demonstrated in the next sections -- the results presented here also provide significant novel insights for the area of first-order modal logics \cite{Gabbay2005-GABMML-2}.
As indicated by our naming, the subcase of \SmonCtwo where the only standpoint expression used is $*$ coincides with the monodic fragment of modal counting two-variable FO with a S5 modal operator.
While it has been observed earlier that restricting to the monodic setting is crucial for maintaining decidability in non-trivial combinations of FO fragments with modalities of varying kinds \cite{WolterZ01},
existing decidability results explicitly exclude FO fragments with equality or function symbols, which are notoriously harder, leaving such cases as an open question. We transcend this boundary, since  \Ctwo{} supports equality and unary functions (via axiomatising binary predicates as functional), and beyond mere decidability, we establish tight complexity bounds.

\pagebreak

\section{Application to Ontology Languages}\label{sec:application-ontology-languages}

We now show that adding monodic standpoints to popular ontology languages does not negatively affect the computational complexity of standard reasoning tasks. To this end, we begin by adding monodic standpoints to the description logic $\mathcal{ALCOIQB}^\mathsf{Self}$, and then we show how we can also accommodate role chain axioms, thus obtaining monodic standpoint $\mathcal{SHOIQB}_s$ and $\mathcal{SROIQB}_s$, which subsume the W3C ontology standards OWL~1 and OWL~2~DL respectively.\footnote{The less mainstream letter $\mathcal{B}$ in the DL names refers to \emph{boolean role constructors}, where $\mathcal{B}_s$ denotes boolean role constructors over simple roles only (see, e.g., \citeauthor{DBLP:conf/jelia/RudolphKH08} \citeyear{DBLP:conf/jelia/RudolphKH08}). This modelling feature is not available in OWL~1 or OWL~2~DL.} For the following, familiarity with description logics \cite{BHLSintroDL,RudolphDLfoundations} will be very helpful.

\subsection{Monodic Standpoint $\mathcal{ALCOIQB}^\mathsf{Self}$}

We first introduce $\mathbb{S}^\mathrm{mon}_{\mathcal{ALCOIQB}^\mathsf{Self}}$ obtained by enhancing the description logic $\mathcal{ALCOIQB}^\mathsf{Self}$ by monodic standpoints. Just like \Ctwo{} \FOSL, $\mathbb{S}^\mathrm{mon}_{\mathcal{ALCOIQB}^\mathsf{Self}}$ is based on a signature $\tuple{\Preds, \Consts, \Stands}$ where $\Preds$ only contains unary and binary predicates, also referred to as \emph{concept names} and \emph{role names}, respectively.
Based on these, we define the set $\mathbf{E}_\mathrm{rol}$ of \emph{role expressions}
$$ R, R' ::= \pred{R} \mid \pred{R}^- \mid \neg R \mid R \cap R'$$
with $\pred{R} \in \Preds$ binary, and the set $\mathbf{E}_\mathrm{con}$ of \emph{concept expressions}
$$ C,D ::= \pred{A} \mid \neg C \mid \{o\} \mid C \sqcap D \mid {\geqslant} n R.C \mid \exists R.\mathsf{Self} \mid \standd{\ste} C $$
with $\pred{A} \in \Preds$ unary, $o \in \Consts$, $n\in \mathbb{N}$, $\ste\in \StandExps$ (see \Cref{def:FOSLsyntax}).
Finally the set of $\mathbb{S}^\mathrm{mon}_{\mathcal{ALCOIQB}^\mathsf{Self}}$ sentences is defined by
$$ \phi, \psi ::= C \sqsubseteq D \mid \neg \phi \mid \phi \wedge \psi \mid \standd{\ste} \phi.$$
We 
introduce $\mathbb{S}^\mathrm{mon}_{\mathcal{ALCOIQB}^\mathsf{Self}}$ with a minimalistic syntax, but note that all the usual description logic constructs can be introduced as syntactic sugar. For example, we obtain $\bot$ as $\pred{A} \sqcap \neg \pred{A}$ and $\top$ as $\neg \bot$; we may write $\exists R.C$ instead of ${\geqslant} 1 R.C$ and also $\forall R.C$ instead of $\neg {\geqslant} 1 R.\neg C$; last not least we may write $\standb{\ste} C$ to denote $\neg\standd{\ste}\neg C$. We also remind the reader that other usual axiom types 
can all be rewritten into statements of the form $C \sqsubseteq D$ (referred to as \emph{general concept inclusions}, short: \emph{GCIs}) in the presence of \emph{nominals} (i.e., expressions of the form $\{o\}$) and role expressions.
Following DL naming conventions, a $\mathbb{S}^\mathrm{mon}_{\mathcal{ALCOIQB}^\mathsf{Self}}$ sentence will be called  a \emph{TBox} if it is a conjunction of GCIs.

For later discussions, we single out some fragments of $\mathbb{S}^\mathrm{mon}_{\mathcal{ALCOIQB}^\mathsf{Self}}$: We obtain $\mathbb{S}^\mathrm{mon}_{\mathcal{ALCOIF}}$ by excluding $\neg$ and $\cap$ from role expressions as well as disallowing concept expressions that use $\mathsf{Self}$ or ${\geqslant}k$ for $k \geq 2$, with the notable exception of axioms of the specific form $\top \sqsubseteq \neg {\geqslant} 2 \pred{F}.\top$ stating the functionality for binary predicates $\pred{F}$, which are then often abbreviated by $\mathit{func}(\pred{F})$.  We obtain $\mathbb{S}^\mathrm{mon}_{\mathcal{ALCO}}$ from $\mathbb{S}^\mathrm{mon}_{\mathcal{ALCOIF}}$ by disallowing role expressions of the form $\pred{R}^-$ (known as \emph{inverses}), and functionality axioms.

The semantics of standpoint-enhanced description logics is usually provided in a model-theoretic way using standpoint structures as in \Cref{def:semantics} \cite{sententialFOLandOWL}. For space reasons, we will instead define the semantics by directly providing a translation into \SmonCtwo. To justify this ``shortcut'' we point out that said translation truthfully reflects the model-theoretic semantics of all earlier described standpoint-enhanced DLs and that existing translations from standpoint-free DLs to plain \Ctwo{} \cite{Kazakov08} naturally arise as a special case of ours.  

The translation of a $\mathbb{S}^\mathrm{mon}_{\mathcal{ALCOIQB}^\mathsf{Self}}$ sentence $\phi$ into a \SmonCtwo sentence is obtained by replacing every GCI $C \sqsubseteq D$ inside $\phi$ by
$\forall x. \big(\mathsf{ctrans}(x,C) \to \mathsf{ctrans}(x,D)\big)$, where $\mathsf{ctrans}: \{x,y\}\times \mathbf{E}_\mathrm{con} \to \mathbb{S}_{\Ctwo}$ is inductively defined:
\begin{align*}
		\mathsf{ctrans}(z,\pred{A})    & = \pred{A}(z) \\[-2pt]
		\mathsf{ctrans}(z,\neg C)      & = \neg \mathsf{ctrans}(z, C)\\[-2pt]
		\mathsf{ctrans}(z,\{o\})       & = z\dot{=}o \\[-2pt]
		\mathsf{ctrans}(z,C \sqcap D)    & = \mathsf{ctrans}(z, C) \wedge \mathsf{ctrans}(z, D) \\[-2pt]
		\mathsf{ctrans}(x,{\geqslant} n R.C)    & = \exists^{\geq n} y. \mathsf{rtrans}(x,y,R) \wedge \mathsf{ctrans}(y, C) \\[-2pt]
		\mathsf{ctrans}(y,{\geqslant} n R.C)    & = \exists^{\geq n} x. \mathsf{rtrans}(y,x,R) \wedge \mathsf{ctrans}(x, C) \\[-2pt]
		\mathsf{ctrans}(z,\exists R.\mathsf{Self})    & = \mathsf{rtrans}(z,z,R) \\[-2pt]
		\mathsf{ctrans}(z,\standd{\ste} C)    & = \exists z. \standd{\ste} \mathsf{ctrans}(z, C),
\end{align*}
using $\mathsf{rtrans}: \{x,y\}\times \{x,y\}\times \mathbf{E}_\mathrm{rol} \to \mathbb{S}_{\Ctwo}$ defined by
\begin{align*}
		\mathsf{rtrans}(z,z',\pred{R})    & = \pred{R}(z,z') \\[-2pt]
		\mathsf{rtrans}(z,z',\pred{R}^-)    & = \pred{R}(z',z) \\[-2pt]
		\mathsf{rtrans}(z,z',\neg R)      & = \neg \mathsf{rtrans}(z,z',R)\\[-2pt]
		\mathsf{rtrans}(z,R \cap R')    & = \mathsf{rtrans}(z,z',R) \wedge \mathsf{rtrans}(z,z',R). 
\end{align*}
It can be readily checked that the translation described is computable in polytime (hence polynomial in output) and indeed yields a \SmonCtwo sentence. Therefore, and in view of the fact that satisfiability is already \NExpTime-hard for the standpoint-free sublogic $\mathcal{ALCOIF}$ \cite{Tobies2000}, we obtain the following tight complexity bounds.

\begin{theorem}
Checking satisfiability of $\mathbb{S}^\mathrm{mon}_{\mathcal{ALCOIQB}^\mathsf{Self}}$ sentences is \NExpTime-complete.
\end{theorem}

\subsection{Adding Role Chain Axioms}\label{sec:addingrolechain}

In order to fully cover the web ontology languages OWL~1 and OWL~2~DL, we need to extend our formalism by so-called \emph{role chain axioms}, arriving at the description logics $\mathcal{SHOIQB}_s$ (when allowing just role chain axioms expressing transitivity such as $\mathtt{FriendOf} \circ \mathtt{FriendOf} \sqsubseteq \mathtt{FriendOf}$)\\ or $\mathcal{SROIQB}_s$ (when admitting more complex forms like $\mathtt{FriendOf} \circ \mathtt{EnemyOf} \sqsubseteq \mathtt{EnemyOf}$), respectively. Luckily, by combining known standpoint encoding tricks \cite{sententialFOLandOWL} and removal techniques for role-chain axioms \cite{Kazakov08,DemriN05} with some novel ideas, it is possible to translate $\mathbb{S}^\mathrm{mon}_{\mathcal{SHOIQB}_s}$ and $\mathbb{S}^\mathrm{mon}_{\mathcal{SROIQB}_s}$ sentences back into $\mathbb{S}^\mathrm{mon}_{\mathcal{ALCOIQB}^\mathsf{Self}}$. For $\mathbb{S}^\mathrm{mon}_{\mathcal{SHOIQB}_s}$, the translation is polynomial, for $\mathbb{S}^\mathrm{mon}_{\mathcal{SROIQB}_s}$ exponential. 


\begin{theorem}
Checking satisfiability of $\mathbb{S}^\mathrm{mon}_{\mathcal{SHOIQB}_s}$ sentences is \NExpTime-complete.
Checking satisfiability of $\mathbb{S}^\mathrm{mon}_{\mathcal{SROIQB}_s}$ sentences is \NTwoExpTime-complete.  
\end{theorem}

Therein, the hardness part for $\mathbb{S}^\mathrm{mon}_{\mathcal{SROIQB}_s}$ follows from the known \NTwoExpTime hardness of its fragment $\mathcal{SROIQ}$ \cite{Kazakov08}.   
This finishes our argument that adding monodic Standpoints to OWL~1 and OWL~2 does not increase complexity of standard reasoning tasks.

\pagebreak

\section{Nominals Cause Trouble}\label{sec:features-causing-trouble}

We finish our considerations by providing two results that provide some context for our main results and support the intuition (cf. \citeauthor{monodicEL} (\citeyear{monodicEL}) as well as \citeauthor{monodicSHIQ} (\citeyear{monodicSHIQ})) that the interplay of nominals and standpoint modalities is particularly troublesome for reasoning.
To this end, we will use the tiling problem in two variations, which we introduce next.

A \emph{tiling system} $\mathbb{T}=\langle k,H,V\rangle$ consists of a number $k \in \mathbb{N}$ indicating the number of tiles, and horziontal and vertical compatibility relations $H,V \subseteq \{1,\ldots,k\}\times\{1,\ldots,k\}$. 
For a downward-closed set $S \subseteq \mathbb{N}$ of natural numbers, 
a $\mathbb{T}$-\emph{tiling} of $S \times S$ with initial condition $\langle t_0,\ldots t_n\rangle \in \{1,\ldots,k\}^n$ for some $n \in S$ is a mapping $\mathsf{tile}:S \times S \to \{1,\ldots, k\}$ such that $\mathsf{tile}(i,0) = t_i$ for $i \in \{1,\ldots,n\}$, and for all $i\in S$ with $i+1 \in S$ and all $j \in S$ holds $(\mathsf{tile}(i,j),\mathsf{tile}(i+1,j)) \in H$ as well as $(\mathsf{tile}(j,i),\mathsf{tile}(j,i+1)) \in V$.
We recall the following:
\begin{itemize}
\item
There is a tiling system $\mathbb{T}_\mathrm{exp}$ such that the following problem is \NExpTime-hard:
Given an initial condition of size $n$, is there a corresponding $\mathbb{T}_\mathrm{exp}$-tiling of $\{0,\ldots,2^n-1\} \times \{0,\ldots,2^n-1\}$?
\item 
There is a tiling system $\mathbb{T}_\mathrm{und}$ such that the following problem is undecidable:
Given an initial condition of size $n$, is there a corresponding $\mathbb{T}_\mathrm{und}$-tiling of $\mathbb{N}\times \mathbb{N}$?
\end{itemize}

\subsection{\NExpTime Hardness for $\mathcal{ALCO}$ TBoxes}

From prior works, it is known that monodic standpoint $\mathcal{SHIQ}$, a sublogic of $\mathbb{S}^\mathrm{mon}_{\mathcal{SHOIQB}_s}$ has an \ExpTime-complete satisfiability problem \cite{monodicSHIQ}, which means that the complexity of $\mathcal{SHIQ}$ is unaltered if monodic standpoints are added.
Two other popular \ExpTime-complete Sub-DLs of $\mathcal{SHOIQB}_s$ (incomparable to $\mathcal{SHIQ}$) are $\mathcal{SHIO}$ and $\mathcal{SHOQ}$ \cite{SHIO,SHOQ}. 
This poses the question if adding monodic standpoints to these DLs preserves \ExpTime reasoning, like it does for $\mathcal{SHIQ}$. 



Interestingly, we can answer this question in the negative (unless $\NExpTime=\ExpTime$) and identify nominals as the joint cause by showing that satisfiability even of monodic standpoint TBoxes in $\mathcal{ALCO}$ (a rather restricted sublogic of both $\mathcal{SHIO}$ and $\mathcal{SHOQ}$) is already NExpTime hard.

To this end, we provide a polynomial reduction from the first of the two above tiling problems to the satisfiability problem of a $\mathbb{S}^\mathrm{mon}_\mathcal{ALCO}$ TBox of size polynomial in $n$, using just one nominal concept $\{o\}$.
We use atomic concepts $\pred{T}_1,\ldots,\pred{T}_k$ for the $k$ tiles
and atomic concepts $\pred{X}_1,\ldots, \pred{X}_{n}$ as well as $\pred{Y}_1,\ldots, \pred{Y}_{n}$ to encode x- and y-coordinates in binary. First, we declare all these concepts as ``almost rigid'': they hold uniformly across precisifications for all elements but $o$.
$$
\neg\{o\} \sqcap \pred{T}_\ell \sqsubseteq {\Box}_* \pred{T}_\ell
\quad
\neg\{o\} \sqcap \pred{X}_i \sqsubseteq {\Box}_* \pred{X}_i
\quad
\neg\{o\} \sqcap \pred{Y}_i \sqsubseteq {\Box}_* \pred{Y}_i
$$
Above and below, we let $i$ range from $1$ to $n$ and let $\ell$ range from $1$ to $k$.
Next, we ensure that, in every precisification, every non-$o$ element with x-coordinate (y-coordinate) smaller than $2^n-1$ has a horizontal (vertical) neighbour with that coordinate incremented and the same y-coordinate (x-coordinate). We let $j$ range from $1$ to $i-1$.
$$
\begin{array}{r@{\ }l@{\ }l}
\neg\{o\} \sqcap \bigsqcup_i \neg \pred{X}_i & \sqsubseteq & \exists \pred{H}. \neg\{o\} \\
\pred{X}_{i} \sqcap \neg\pred{X}_j  & \sqsubseteq & \forall \pred{H}.\pred{X}_{i}\\
\neg\pred{X}_{i} \sqcap \neg\pred{X}_j  & \sqsubseteq & \forall \pred{H}.\neg\pred{X}_{i}\\
\pred{X}_{i} \sqcap \bigsqcap_{j} \pred{X}_j & \sqsubseteq & \forall \pred{H}.\neg \pred{X}_{i}\\
\neg \pred{X}_{i} \sqcap \bigsqcap_{j} \pred{X}_j & \sqsubseteq & \forall \pred{H}. \pred{X}_{i}\\
\end{array}
\begin{array}{r@{\ }l@{\ }l}
\neg\{o\} \sqcap \bigsqcup_i \neg \pred{Y}_i & \sqsubseteq & \exists \pred{V}. \neg\{o\} \\
\pred{Y}_{i} \sqcap \neg\pred{Y}_j  & \sqsubseteq & \forall \pred{V}.\pred{Y}_{i}\\
\neg\pred{Y}_{i} \sqcap \neg\pred{Y}_j  & \sqsubseteq & \forall \pred{V}.\neg\pred{Y}_{i}\\
\pred{Y}_{i} \sqcap \bigsqcap_{j} \pred{Y}_j & \sqsubseteq & \forall \pred{V}.\neg \pred{Y}_{i}\\
\neg \pred{Y}_{i} \sqcap \bigsqcap_{j} \pred{Y}_j & \sqsubseteq & \forall \pred{V}. \pred{Y}_{i}\\
\end{array}
$$

\smallskip

$\pred{Y}_i \sqsubseteq \forall \pred{H}.\pred{Y}_i$\quad\!\!\!
$\neg \pred{Y}_i \sqsubseteq \forall \pred{H}.\neg \pred{Y}_i$
\quad\!\!\!
$\pred{X}_i \sqsubseteq \forall \pred{V}.\pred{X}_i$\quad\!\!\!
$\neg \pred{X}_i \sqsubseteq \forall \pred{V}.\neg \pred{X}_i$

\smallskip

We next ensure that there exists a non-$o$ element with x and y set to zero, which together with its horizontal neighbours realises the initial condition $\langle t_0,\ldots t_n\rangle \in \{1,\ldots,k\}^n$.

\smallskip

$\top \sqsubseteq \exists \pred{R}.(\neg\{o\} \sqcap \neg \pred{X}_1 \sqcap \ldots \sqcap \neg \pred{X}_{n} \sqcap \neg \pred{Y}_1 \sqcap \ldots \sqcap \neg \pred{Y}_{n} \sqcap$

$\qquad\ \  \pred{T}_{t_1} \sqcap \forall \pred{H}.(\pred{T}_{t_2} \sqcap \forall\pred{H}.( \ldots (\pred{T}_{t_{n-1}} \sqcap \forall \pred{H}.\pred{T}_{t_n}) \ldots )) $

\smallskip

For every non-$o$ element, there exists some precisification wherein it is $\pred{P}$-linked to 
$o$ and propagates its x- and y- coordinate as well as its tile assignment via this link to $o$.  

\smallskip
$\begin{array}{r@{\ }l@{\ }l}
\neg\{o\} & \sqsubseteq & \Diamond_* \exists \pred{P}.\{o\}\\
\pred{T}_i & \sqsubseteq & \forall \pred{P}.\pred{T}_i
\end{array}$
$\begin{array}{r@{\ }l@{\ }l}
\pred{X}_i & \sqsubseteq & \forall \pred{P}.\pred{X}_i\\
\neg\pred{X}_i & \sqsubseteq & \forall \pred{P}.\neg\pred{X}_i
\end{array}$
$\begin{array}{r@{\ }l@{\ }l}
\pred{Y}_i & \sqsubseteq & \forall \pred{P}.\pred{Y}_i\\
\neg\pred{Y}_i & \sqsubseteq & \forall \pred{P}.\neg\pred{Y}_i
\end{array}$
\smallskip

In every precisification, every non-$o$ element is  $\pred{P}'$-linked to 
$o$ and, should its assigned x- and y-coordinate coincide with those assigned to $o$, then its tile-assignment will coincide with the one of $o$ as well.

\smallskip
$\neg\{o\} \sqsubseteq \exists \pred{P}'.\{o\} $

$\exists \pred{P}'.\pred{T}_\ell \sqcap \bigsqcap_{i} \big((\pred{X}_i \sqcap \exists \pred{P}'.\pred{X}_i) \sqcup (\neg\pred{X}_i \sqcap \exists \pred{P}'.\neg\pred{X}_i)\big)$

$\qquad\ \ \ \ \sqcap\bigsqcap_{i} \big((\pred{Y}_i \sqcap \exists \pred{P}'.\pred{Y}_i) \sqcup (\neg\pred{Y}_i \sqcap \exists \pred{P}'.\neg\pred{Y}_i)\big) \sqsubseteq \pred{T}_\ell$

\smallskip

Note that this way, the tile assignments will be synchronized between all elements carrying the same coordinates.
We finally make sure that in every precisification, every domain element must be assigned a tile. Moreover the $\pred{H}$- and $\pred{V}$-neighbouring pairs of elements must conform with the horizontal and vertical compatibility relation.

\smallskip

$\top \sqsubseteq \pred{T}_1 \sqcup \ldots \sqcup  \pred{T}_k$

$\pred{T}_\ell \sqsubseteq \forall \pred{H}.\neg \pred{T}_{\ell'}$ \quad for $(\ell,{\ell'}) \in  \{1,\ldots,k\}\times\{1,\ldots,k\}  \setminus H$

$\pred{T}_\ell \sqsubseteq \forall \pred{V}.\neg \pred{T}_{\ell'}$ \quad for $(\ell,{\ell'}) \in \{1,\ldots,k\}\times\{1,\ldots,k\} \setminus V$

\smallskip

This finishes the description of the TBox (obtained by taking the conjunction of all the introduced GCIs). We note that these axioms do not enforce the $\pred{H}$ and $\pred{V}$ relation to form a proper grid (in any precisification). Rather, the axioms ensure that for any two horizontally (vertically) neighbouring coordinate pairs, there exists a $\pred{H}$-connected ($\pred{V}$-connected) pair of domain elements carrying said coordinates. Since the tile assignments are rigid (except for $o$) and synchronized over all elements carrrying equal coordinates, this suffices to ensure that satisfiability of our TBox coincides with the existence of a $\mathbb{T}_\mathrm{exp}$-tiling, so we obtain the following theorem.


\begin{theorem}
In any sublogic of $\mathbb{S}^\mathrm{mon}_{\mathcal{SHOIQB}_s}$ that subsumes $\mathbb{S}^\mathrm{mon}_\mathcal{ALCO}$ TBoxes, satisfiability is \NExpTime-complete. 
\end{theorem}

\subsection{Lifting Monodicity Causes Undecidability}

A crucial restriction underlying all logical formalisms that we have considered so far is monodicity, that is, that modal operators can only be put in front of subformulae with at most one free variable. 
The arguably mildest way of lifting monodicity is by imposing that one distinguished binary predicate, say $\rigidp$, must be rigidly interpreted. 
Note that rigidity of a binary predicate $\rigidp$ could be expressed by the FOSL formula $\forall x,y. \big( \rigidp(x,y) \to \allstandb \rigidp(x,y) \big)$, which is not monodic.
By a reduction from the second of the above tiling problems, we show that adding one rigid binary predicate causes undecidability even for a sublogic of $\mathbb{S}^\mathrm{mon}_{\mathcal{SHOIQB}_s}$.

\begin{theorem}
Satisfiability of $\mathbb{S}^\mathrm{mon}_\mathcal{ALCOIF}$ TBoxes with one rigid binary predicate is undecidable, even when using just one nominal and one functionality statement.
\end{theorem}

For space reasons, we only briefly provide a set of GCIs enforcing an $\mathbb{N} \times \mathbb{N}$ grid, noting that a $\mathbb{T}_\mathrm{und}$-tiling on top can be obtained very similarly to the previous case.
Let $\rigidp$ be the distinguished rigid binary predicate, which we use to represent both horizontal and vertical grid connections (distinguishing them through extra unary rigid ``markers'' for even/odd grid rows). Let $\mathit{func}(\pred{Point})$ specify that the ``pointer predicate'' $\pred{Point}$ is functional and
put the following GCIs.

\medskip


$\top \sqsubseteq \exists\rigidp.\Box_*\pred{Even} \sqcap \exists\rigidp.\Box_*\pred{Odd} \sqcap \Diamond_{*} \pred{Pick}$


$\pred{Pick} \sqsubseteq \forall \rigidp.\big(\neg \pred{Even} \sqcup\forall \rigidp.(\neg \pred{Odd} \sqcup \exists \pred{Point}^-.\{o\}) \big) $  

$\pred{Pick} \sqsubseteq \forall \rigidp.\big(\neg \pred{Odd} \sqcup\forall \rigidp.(\neg \pred{Even} \sqcup \exists \pred{Point}^-.\{o\}) \big)$  

\medskip

In a nutshell, these GCIs ensure that every grid element $\delta$ will be $\pred{Pick}$ed in some precisification, and in that precisification the upper neighbour of $\delta$'s right neighbour is forced to coincide with the right neighbour of $\delta$'s upper neighbour, by having both being ``functionally $\pred{Point}$ed to'' from $o$.

We note that this finding contrasts with a positive result by \citeauthor{S5forALCIQrigidroles} (\citeyear{S5forALCIQrigidroles}), according to which -- in our nomenclature -- the satisfiability of TBoxes over $\mathcal{ALCIQ}$ with arbitrarily many rigid roles and one S5 modality allowed to occur in front of concept \emph{and role expressions}, is decidable in \TwoExpTime. Once more, this underlines the previous observation that while counting and inverses go reasonably well with standpoint modalities, nominals do not.

\section{Conclusions and Future Work}\label{sec:conclusions}

We have shown that monodic standpoints can be added to \Ctwo{} without increasing the \NExpTime reasoning complexity. We obtained this result by establishing a polynomial translation into plain \Ctwo, whose justification required rather elaborate model-theoretic constructions and arguments. On one hand, this finding advances the research into first-order modal logics, since our result subsume the case of monodic S5 over \Ctwo{} and even apply to logics with more expressive counting. On the other hand, we showed how the obtained result can be leveraged to prove that very expressive DLs subsuming popular W3C-standardized ontology languages can be endowed with monodic standpoints in a complexity-neutral way. We finally showed that in the presence of nominal concepts, \NExpTime-hardness already arises for much less expressive DLs, and lifting monodicity even incurs undecidability. 

For future work, it would be interesting to investigate the data complexity of our formalism.
Also it would be advantageous to find translations from versions of monodic standpoint OWL into plain OWL rather than \Ctwo{}, since this would allow to deploy existing highly optimized OWL reasoners for standpoint-aware ontological reasoning. While our results show that there are no complexity-theoretic barriers for this, our current translation approach heavily relies on features of \Ctwo{} that are beyond the capabilities of plain OWL.    

\clearpage

\section*{Acknowledgements}
This work is funded by the Agence Nationale de la Recherche, France (ANR) in project ANR-25-CE23-2478-01 (SPaRK), by
the Federal Ministry of Research, Technology and Space, Germany (BMFTR) in \href{https://www.scads.de}{ScaDS.AI},
and by BMFTR and DAAD in project 57616814 (SECAI).

\bibliographystyle{kr}
\bibliography{bib/references-fol,bib/references}

\begin{thebibliography}{}

\bibitem[\protect\citeauthoryear{Artale, Lutz, and Toman}{2007}]{S5forALCIQrigidroles}
Artale, A.; Lutz, C.; and Toman, D.
\newblock 2007.
\newblock A description logic of change.
\newblock In Veloso, M.~M., ed., {\em Proceedings of the 20th International Joint Conference on Artificial Intelligence},  218--223.
\newblock AAAI Press.

\bibitem[\protect\citeauthoryear{Baader \bgroup et al\mbox.\egroup }{2017}]{BHLSintroDL}
Baader, F.; Horrocks, I.; Lutz, C.; and Sattler, U.
\newblock 2017.
\newblock {\em An Introduction to Description Logic}.
\newblock Cambridge University Press.

\bibitem[\protect\citeauthoryear{Bao \bgroup et al\mbox.\egroup }{2009}]{owl2-overview}
Bao, J.; Calvanese, D.; Grau, B.~C.; Dzbor, M.; Fokoue, A.; Golbreich, C.; Hawke, S.; Herman, I.; Hoekstra, R.; Horrocks, I.; Kendall, E.; Krötzsch, M.; Lutz, C.; McGuinness, D.~L.; Motik, B.; Pan, J.; Parsia, B.; Patel-Schneider, P.~F.; Rudolph, S.; Ruttenberg, A.; Sattler, U.; Schneider, M.; Smith, M.; Wallace, E.; Wu, Z.; and Zimmermann, A.
\newblock 2009.
\newblock {OWL~2 Web Ontology Language: Document Overview. W3C Recommendation}.
\newblock \url{http://www.w3.org/TR/owl2-overview/}.

\bibitem[\protect\citeauthoryear{Benedikt, Kostylev, and Tan}{2020}]{ultimatelyperiodic}
Benedikt, M.; Kostylev, E.~V.; and Tan, T.
\newblock 2020.
\newblock Two variable logic with ultimately periodic counting.
\newblock In Czumaj, A.; Dawar, A.; and Merelli, E., eds., {\em Proceedings of the 47th International Colloquium on Automata, Languages and Programming}, volume 168 of {\em LIPIcs},  112:1--112:16.
\newblock Schloss Dagstuhl - Leibniz-Zentrum f{\"{u}}r Informatik.

\bibitem[\protect\citeauthoryear{Demri and de Nivelle}{2005}]{DemriN05}
Demri, S., and de~Nivelle, H.
\newblock 2005.
\newblock Deciding regular grammar logics with converse through first-order logic.
\newblock {\em Journal of Logic Language and Information} 14(3):289--329.

\bibitem[\protect\citeauthoryear{Gabbay \bgroup et al\mbox.\egroup }{2005}]{Gabbay2005-GABMML-2}
Gabbay, D.~M.; Kurucz, A.; Wolter, F.; and Zakharyaschev, M.
\newblock 2005.
\newblock Many-dimensional modal logics: Theory and applications.
\newblock {\em Studia Logica} 81(1):147--150.

\bibitem[\protect\citeauthoryear{Glimm, Horrocks, and Sattler}{2008}]{SHOQ}
Glimm, B.; Horrocks, I.; and Sattler, U.
\newblock 2008.
\newblock Unions of conjunctive queries in {SHOQ}.
\newblock In Brewka, G., and Lang, J., eds., {\em Proceedings of the 11th International Conference on Principles of Knowledge Representation and Reasoning},  252--262.
\newblock {AAAI} Press.

\bibitem[\protect\citeauthoryear{{G\'{o}mez \'{A}lvarez} and Rudolph}{2021}]{gomez2021standpoint}
{G\'{o}mez \'{A}lvarez}, L., and Rudolph, S.
\newblock 2021.
\newblock Standpoint logic: Multi-perspective knowledge representation.
\newblock In Neuhaus, F., and Brodaric, B., eds., {\em Proceedings of the 12th International Conference on Formal Ontology in Information Systems}, volume 344 of {\em FAIA},  3--17.
\newblock IOS Press.

\bibitem[\protect\citeauthoryear{{Gómez \smash{Á}lvarez} and Rudolph}{2024}]{monodicSHIQ}
{Gómez \smash{Á}lvarez}, L., and Rudolph, S.
\newblock 2024.
\newblock Reasoning in {SHIQ} with axiom- and concept-level standpoint modalities.
\newblock In Marquis, P.; Ortiz, M.; and Pagnucco, M., eds., {\em Proceedings of the 21st International Conference on Principles of Knowledge Representation and Reasoning},  383--393.
\newblock ijcai.org.

\bibitem[\protect\citeauthoryear{{Gómez \smash{Á}lvarez}, Rudolph, and Strass}{2022}]{sententialFOLandOWL}
{Gómez \smash{Á}lvarez}, L.; Rudolph, S.; and Strass, H.
\newblock 2022.
\newblock How to agree to disagree - managing ontological perspectives using standpoint logic.
\newblock In Sattler, U.; Hogan, A.; Keet, C.~M.; Presutti, V.; Almeida, J. P.~A.; Takeda, H.; Monnin, P.; Pirr{\`{o}}, G.; and d'Amato, C., eds., {\em Proceedings of the 21st International Semantic Web Conference}, volume 13489 of {\em LNCS},  125--141.
\newblock Springer.

\bibitem[\protect\citeauthoryear{{Gómez \smash{Á}lvarez}, Rudolph, and Strass}{2023a}]{monodicELplus}
{Gómez \smash{Á}lvarez}, L.; Rudolph, S.; and Strass, H.
\newblock 2023a.
\newblock Pushing the boundaries of tractable multiperspective reasoning: {A} deduction calculus for standpoint {EL+}.
\newblock In Marquis, P.; Son, T.~C.; and Kern{-}Isberner, G., eds., {\em Proceedings of the 20th International Conference on Principles of Knowledge Representation and Reasoning},  333--343.
\newblock ijcai.org.

\bibitem[\protect\citeauthoryear{{Gómez \smash{Á}lvarez}, Rudolph, and Strass}{2023b}]{monodicEL}
{Gómez \smash{Á}lvarez}, L.; Rudolph, S.; and Strass, H.
\newblock 2023b.
\newblock Tractable diversity: Scalable multiperspective ontology management via standpoint {EL}.
\newblock In Elkind, E., ed., {\em Proceedings of the Thirty-Second International Joint Conference on Artificial Intelligence},  3258--3267.
\newblock ijcai.org.

\bibitem[\protect\citeauthoryear{Hladik and Model}{2004}]{SHIO}
Hladik, J., and Model, J.
\newblock 2004.
\newblock Tableau systems for {SHIO} and {SHIQ}.
\newblock In Haarslev, V., and M{\"{o}}ller, R., eds., {\em Proceedings of the 2004 International Workshop on Description Logics}, volume 104 of {\em {CEUR} Workshop Proceedings}.
\newblock CEUR-WS.org.

\bibitem[\protect\citeauthoryear{Horrocks, Kutz, and Sattler}{2006}]{DBLP:conf/kr/HorrocksKS06}
Horrocks, I.; Kutz, O.; and Sattler, U.
\newblock 2006.
\newblock The even more irresistible {SROIQ}.
\newblock In Doherty, P.; Mylopoulos, J.; and Welty, C.~A., eds., {\em Proceedings 10th International Conference on Principles of Knowledge Representation and Reasoning},  57--67.
\newblock {AAAI} Press.

\bibitem[\protect\citeauthoryear{Kazakov}{2008}]{Kazakov08}
Kazakov, Y.
\newblock 2008.
\newblock {RIQ} and {SROIQ} are harder than {SHOIQ}.
\newblock In Brewka, G., and Lang, J., eds., {\em Proceedings of the 11th International Conference on Principles of Knowledge Representation and Reasoning},  274--284.
\newblock {AAAI} Press.

\bibitem[\protect\citeauthoryear{Kurucz, Wolter, and Zakharyaschev}{2023}]{DBLP:conf/kr/KuruczWZ23}
Kurucz, A.; Wolter, F.; and Zakharyaschev, M.
\newblock 2023.
\newblock Definitions and (uniform) interpolants in first-order modal logic.
\newblock In Marquis, P.; Son, T.~C.; and Kern{-}Isberner, G., eds., {\em Proceedings of the 20th International Conference on Principles of Knowledge Representation and Reasoning},  417--428.
\newblock ijcai.org.

\bibitem[\protect\citeauthoryear{Pratt{-}Hartmann}{2005}]{Pratt-Hartmann05}
Pratt{-}Hartmann, I.
\newblock 2005.
\newblock Complexity of the two-variable fragment with counting quantifiers.
\newblock {\em Journal of Logic, Language, and Information} 14(3):369--395.

\bibitem[\protect\citeauthoryear{Rudolph, Kr{\"{o}}tzsch, and Hitzler}{2008}]{DBLP:conf/jelia/RudolphKH08}
Rudolph, S.; Kr{\"{o}}tzsch, M.; and Hitzler, P.
\newblock 2008.
\newblock Cheap {Boolean} role constructors for description logics.
\newblock In H{\"{o}}lldobler, S.; Lutz, C.; and Wansing, H., eds., {\em Proceedings of the 11th European Conference on Logics in Artificial Intelligence}, volume 5293 of {\em LNCS},  362--374.
\newblock Springer.

\bibitem[\protect\citeauthoryear{Rudolph}{2011}]{RudolphDLfoundations}
Rudolph, S.
\newblock 2011.
\newblock Foundations of description logics.
\newblock In Polleres, A.; d'Amato, C.; Arenas, M.; Handschuh, S.; Kroner, P.; Ossowski, S.; and Patel{-}Schneider, P.~F., eds., {\em Reasoning Web. Semantic Technologies for the Web of Data - 7th International Summer School 2011, Tutorial Lectures}, volume 6848 of {\em Lecture Notes in Computer Science},  76--136.
\newblock Springer.

\bibitem[\protect\citeauthoryear{Tobies}{2000}]{Tobies2000}
Tobies, S.
\newblock 2000.
\newblock The complexity of reasoning with cardinality restrictions and nominals in expressive description logics.
\newblock {\em Journal of Artificial Intelligence Research} 12:199--217.

\bibitem[\protect\citeauthoryear{Wolter and Zakharyaschev}{2001}]{WolterZ01}
Wolter, F., and Zakharyaschev, M.
\newblock 2001.
\newblock Decidable fragments of first-order modal logics.
\newblock {\em Journal of Symbolic Logic} 66(3):1415--1438.

\end{thebibliography}
\appendix

\onecolumn
\section{Proofs of Section 2}

\begin{theorem-repeat}[of \Cref{theorem:fosl-equisatisfiable-to-S5}]
For any \FOSL formula $\phi$, an equisatisfiable S5 \FOSL formula $\mathsf{S5}(\phi)$ can be computed in polynomial time. The transformation preserves $\mathcal{C}^2$-ness and monodicity. 
\end{theorem-repeat}

\begin{proof}
Let $\phi$ be a monodic $\mathcal{C}^2$ \FOSL formula based on a signature $\tuple{\Preds, \Consts, \Stands}$. We show that for any formula $\phi$, the formula $\trans(\phi)$, based on the signature $\tuple{\Preds\cup\Stands, \Consts, \emptyset}$ is equisatisfiable and preserves $\mathcal{C}^2$-ness and monodicity. 
The function $\trans$ replaces every $\standd \psi$ by $\allstandd (e \land \psi)$, introducing one nullary predicate for every standpoint symbol and translating set expressions for standpoints into boolean expressions. The function $\trans$ is defined as follows

\begin{align}
		\trans(\pred{P}(t_1,\ldots,t_k)) & = \pred{P}(t_1,\ldots,t_k)                                                           \\
		\trans(t_1 \,\dot{=}\, t_2) & = t_1 \,\dot{=}\, t_2                                                           \\
		\trans(\lnot\psi)                & =\lnot\trans(\psi)                                                                 \\
		\trans(\psi_1{\,\land\,}\psi_2)  & = \trans(\psi_1){\,\land\,}\trans(\psi_2)                                    \\
		\trans(\forall x\psi)            & = \forall x(\trans(\psi))                                                         \\
		\trans(\standd{\ste}\psi)     & = \standd{*}(\transE(\ste) \land \trans(\psi))
	\end{align}

    Therein, $\transE$ implements the semantics of standpoint expressions, providing for
each expression \mbox{$\ste\in\StandExps$} a propositional formula $\transE(\ste)$ as follows\/:%
\begin{align}
	\transE(\sts)                  & = \sts                                       \\
	\transE(\ste_1\cup\ste_2)      & = \transE(\ste_1)\lor\transE(\ste_2)      \\
	\transE(\ste_1\cap\ste_2)      & = \transE(\ste_1)\land\transE(\ste_2)     \\
	\transE(\ste_1\setminus\ste_2) & = \transE(\ste_1)\land\neg\transE(\ste_2)
\end{align}

Let us show that $\phi$ and $\trans(\phi)$ are equisatisfiable. Let $\kstruct=\tuple{\Dom,\Precs,\sigma,\gamma}$ be a \FOSL structure and let $\kstruct'=\tuple{\Dom,\Precs,\sigma',\gamma'}$ be such that 
$\sigma'$ is the empty function (because there are no standpoint symbols -- note that $\sigma'$ when lifted to standpoint expressions still maps $*$ to $\Precs$) and 
$\gamma'$ is the extension of $\gamma$ to the additional unary predicates of the signature of $\trans(\phi)$ such that $\st^{\gamma'(\pr)}={\mathbf{t}}$ iff $\pr\in\sigma(\st)$.

\noindent We begin by showing, on the structure of $\ste$, that $\pr\in\sigma(\ste)$ iff $\kstruct',\pr\models\transE(\ste)$. 

\begin{description}
    \item [$\ste$ is of the form $\st$]. By construction  $\kstruct',\pr\models\st$ iff $\pr\in\sigma(\st)$.
    \item [$\ste$ is of the form  $\ste_1\cup\ste_2$]. 
    \begin{itemize}[leftmargin=1em,label=-]
        \item $\pr\in\sigma(\ste_1)\cup\sigma(\ste_2)$ iff $\pr\in\sigma(\ste_1)$ or $\pr\in\sigma(\ste_2)$ by the semantics
        \item $\pr\in\sigma(\ste_1)$ or $\pr\in\sigma(\ste_2)$ iff $\kstruct',\pr\models\transE(\ste_1)$ or $\kstruct',\pr\models\transE(\ste_2)$ by induction
        \item $\kstruct',\pr\models\transE(\ste_1)$ or $\kstruct',\pr\models\transE(\ste_2)$ iff $\kstruct',\pr\models\transE(\ste_1)\lor\transE(\ste_2)$ by the semantics
    \end{itemize}
    
    \item [$\ste$ is of the form $\ste_1\cap\ste_2$]. 
    \begin{itemize}[leftmargin=1em,label=-]
        \item $\pr\in\sigma(\ste_1)\cap\sigma(\ste_2)$ iff $\pr\in\sigma(\ste_1)$ and $\pr\in\sigma(\ste_2)$ by the semantics
        \item $\pr\in\sigma(\ste_1)$ and $\pr\in\sigma(\ste_2)$ iff $\kstruct',\pr\models\transE(\ste_1)$ and $\kstruct',\pr\models\transE(\ste_2)$ by induction
        \item $\kstruct',\pr\models\transE(\ste_1)$ and $\kstruct',\pr\models\transE(\ste_2)$ iff $\kstruct',\pr\models\transE(\ste_1)\land\transE(\ste_2)$ by the semantics
    \end{itemize}
    \item [$\ste$ is of the form  $\ste_1\setminus\ste_2$].
    \begin{itemize}[leftmargin=1em,label=-]
        \item $\pr\in\sigma(\ste_1)\setminus\sigma(\ste_2)$ iff $\pr\in\sigma(\ste_1)$ and $\pr\notin\sigma(\ste_2)$ by the semantics
        \item $\pr\in\sigma(\ste_1)$ and $\pr\notin\sigma(\ste_2)$ iff $\kstruct',\pr\models\transE(\ste_1)$ and $\kstruct',\pr\nvDash\transE(\ste_2)$ by induction
        \item $\kstruct',\pr\models\transE(\ste_1)$ and $\kstruct',\pr\nvDash\transE(\ste_2)$ iff $\kstruct',\pr\models\transE(\ste_1)\land\neg\transE(\ste_2)$ by the semantics
    \end{itemize}
\end{description}

Now, we show inductively on the structure of $\phi$ that for all $\pr\in\Precs$, $\kstruct,\pr\models\phi$ iff $\kstruct',\pr\models\trans(\phi)$.
The cases (1)-(4) are trivial, so we focus on case (5).
Thus we show that $\kstruct,\pr\models\standd{\ste}\psi$ iff $\kstruct',\pr\models\standd{*}(\transE(\ste) \land \trans(\psi))$. 
\begin{itemize}[leftmargin=1em,label=-]
    \item $\kstruct,\pr\models\standd{\st}\psi$ iff $\kstruct,\pr'\models\psi$ for some $\pr'\in\sigma(\ste)$ by the semantics
    \item $\kstruct,\pr'\models\psi$ iff $\kstruct',\pr'\models\trans(\psi)$ by induction 
    \item $\pr'\in\sigma(\ste)$ iff $\kstruct',\pr'\models\transE(\ste)$, by the proof of $\transE$
    \item $\kstruct',\pr'\models\trans(\psi)$ and $\kstruct',\pr'\models\transE(\ste)$ iff $\kstruct',\pr\models\allstandd(\transE(\ste)\land\trans(\psi))$ by the semantics
\end{itemize}

%

Once equisatisfiability is shown, a routine inspection of the translation ensures that it preserves C2-ness and monodicity, it can be done in polynomial time and its output is of polynomial size. This translation is similar in spirit to previous translations of ours and to Wolter's
\end{proof}

\begin{theorem-repeat}[of \Cref{lemma:permutations-model-the-same}]
Let $\phi$ be a frugal \SmonCtwo formula and $\modelPR$ the \emph{$\Preds_{\rigidp}$-stable permutational closure} of a standpoint structure $\kstruct$. Let $\varassign'=f'\circ f^{-1}\circ\varassign$. Then, 
$$\modelPR,\!(\pr,f),\!\varassign \models \phi \Longleftrightarrow \modelPR,\!(\pr,f'),\!\varassign' \models \phi$$
\end{theorem-repeat}

\begin{proof}
    We prove this by induction on the structure of $\phi$.
    \begin{description}
    \item[$\phi$ is $\pred{P}(z)$]. Follows from the construction of $\modelPR$ (Def. \ref{def:perm-closure})
    \item[$\phi$ is $\pred{P}(z,z')$]. idem
    \item[$\phi$ is $z \,\dot{=}\, z'$]. idem
    \item[$\phi$ is $\neg\psi$]. 
    \begin{itemize}[leftmargin=1em,label=-]
        \item $\modelPR,\!(\pr,f),\!\varassign \models \neg\psi$ iff $\modelPR,\!(\pr,f),\!\varassign \nvDash \psi$ by the semantics
        \item $\modelPR,\!(\pr,f),\!\varassign \nvDash \psi$ iff $\modelPR,\!(\pr,f'),\!\varassign' \nvDash \psi$ by induction
        \item $\modelPR,\!(\pr,f'),\!\varassign' \nvDash \psi$ iff $\modelPR,\!(\pr,f'),\!\varassign' \models \neg \psi$ by the semantics
    \end{itemize}
    \item[$\phi$ is $\psi_1\land\psi_2$]. Similarly easy.
    \item[$\phi$ is $\exists^{\lhd n}  z.\psi$]. 
    \begin{itemize}[leftmargin=1em,label=-]
        \item $\modelPR,\!(\pr,f),\!\varassign \models \exists^{\lhd n}  z.\psi$ iff $|\{ \delta \mid \modelPR,\!(\pr,f),\!\varassign_{\set{z\mapsto\de}}\models\psi \}|\lhd n$ by the semantics
        \item For all $\delta\in\Delta$, $\modelPR,\!(\pr,f),\!\varassign_{\set{z\mapsto\de}}\models\psi$ iff $\modelPR,\!(\pr,f'),\!\varassign'_{\set{z\mapsto f'\circ f^{-1}(\de)}}\models\psi$ by induction
        \item Since $f'\circ f^{-1}$ is a bijection, $|\{ \delta \mid \modelPR,\!(\pr,f),\!\varassign_{\set{z\mapsto\de}}\models\psi \}|\lhd n$ iff $|\{ f'\circ f^{-1}(\delta) \mid \modelPR,\!(\pr,f'),\!\varassign'_{\set{z\mapsto f'\circ f^{-1}(\de)}}\models\psi \}|\lhd n$
        \item $|\{ f'\circ f^{-1}(\delta) \mid \modelPR,\!(\pr,f'),\!\varassign'_{\set{z\mapsto f'\circ f^{-1}(\de)}}\models\psi \}|\lhd n$ iff $\modelPR,\!(\pr,f'),\!\varassign' \models\exists^{\lhd n}  z.\psi$ by the semantics
        
    \end{itemize}    
    \item[$\phi$ is $\allstandd\psi$]. If $\psi$ has a free variable, let that be $z$. 
    \begin{itemize}[leftmargin=1em,label=-]
        \item $\modelPR,\!(\pr,f),\!\varassign \models \allstandd\psi$ iff there is $(\pr_{\psi},f_{\psi})\in\Precs$ such that $\modelPR,\!(\pr_{\psi},f_{\psi}),\!\varassign\models \psi$ by the semantics
        \item $\modelPR,\!(\pr_{\psi},f_{\psi}),\!\varassign\models \psi$ iff $\modelPR,\!(\pr_{\psi},f'_{\psi}),\!\varassign'\models\psi$ for  $f_{\psi}'=f'\circ f^{-1} \circ f_{\psi}$ by induction
        \item $\modelPR,\!(\pr_{\psi},f'_{\psi}),\!\varassign'\models\psi$ iff $\modelPR,\!(\pr,f'),\!\varassign'\models\allstandd\psi$ by the semantics 
    \end{itemize}
\end{description}
    
\end{proof}

\begin{theorem-repeat}[of \Cref{theorem:sat-in-permutational-closure}]
Let $\phi$ be a satisfiable frugal \SmonCtwo formula over the signature $\tuple{\Preds,\emptyset,\{*\}}$.
Let $Dia_{\phi}$ denote the diamond subformulae of $\phi$ and $\mathit{FreeDia}_{\phi}$ the diamond subformulae with one free variable.
Then there is a standpoint structure $\kstruct=\tuple{\Dom,\Precs,\sigma,\gamma}$ over $\tuple{\Preds \uplus \Preds_{\rigidp},\emptyset,\{*\}}$ with
\begin{itemize}
\item $|\Preds_{\rigidp}| = \ell = |\mathit{FreeDia}_{\phi}|$
\item $|\Precs| \leq |Dia_{\phi}|\cdot 2^{|Dia_{\phi}|} $
\end{itemize}
such that $\modelPR$ is a model of $\phi$. 
\end{theorem-repeat}

\begin{proof}
\noindent We show that if $\phi$ has a model $\kstruct_{o}=\tuple{\Dom,\Precs_{o},\sigma_{o},\gamma_{o}}$, then there is also a structure $\kstruct=\tuple{\Dom,\Precs,\sigma,\gamma}$ of the specified shape such that $\modelPR=\tuple{\Dom,\Precs',\sigma',\gamma'}$ models $\phi$. We assume that $FreeDia_{\phi}=\{\allstandd\phi_1, \ldots , \allstandd\phi_\ell\}$ is linearly ordered and start with some definitions:

\begin{itemize}
    \item $\Preds_{\rigidp}=\{\rigidp_1,\ldots,\rigidp_{\ell}\}$ is a fresh set of predicates 
    \item $f_{id}:\Delta\to\Delta$ denotes the identity function.
\end{itemize}

\medskip
Let $\kstruct_{\rigidp}=\tuple{\Dom,\Precs_{o},\sigma_{o},\gamma_{\rigidp}}$ be the structure based on the signature $\tuple{\Preds \uplus \Preds_{\rigidp},\emptyset,\{*\}}$ and obtained by adding to $\kstruct_{o}$  the extension of $\rigidp_i$ for $i \in \{1,\ldots,\ell\}$ as the set defined by $\{\delta \mid \kstruct_{o},\textcolor{blue}{\pi,} z\mapsto \delta \models \allstandd\phi_i \}$, where $z \in \{x,y\}$ is the one free variable of $\phi_i$, and $\pi\in \Precs_{o}$ is arbitrary (the choice is irrelevant) due to rigidity of $\allstandd\phi_i$.
It is clear that $\kstruct_{\rigidp}$ is a model of $\phi$ iff $\kstruct_{o}$ is a model of $\phi$.

\medskip

\noindent Now, let $\kstruct$ be such that $\sigma$ and $\gamma$ are the restrictions of $\sigma_{o}$ and $\gamma_{\rigidp}$ on $\Precs$, where $\Precs$ is obtained from $\Precs_o$ as follows:
\begin{enumerate}[ label={(P\arabic*)}, ref={P\arabic*}]
    \item If $Dia_{\phi}$ is empty, take an arbitrary precisification $\pr\in\Precs_o$. Otherwise,
    \item For each $\allstandd\psi\in Dia_{\phi}\setminus FreeDia_{\phi}$ such that $\kstruct\models\allstandd\psi$, we take some $\pr\in\Precs_o$ such that $\kstruct,\pr\models\psi$
    \item For each realized $T \subseteq \Preds_{\rigidp}$, pick some $\delta \in \Delta$ that has the $\rigidp$-type $T$, i.e. satisfying
    \begin{itemize}
    \item $\delta \in \pred{E}_j^{\kstruct_{\rigidp}}$ for all $\mathtt{E}_j\in T$ and 
    \item $\delta \notin \pred{E}_j^{\kstruct_{\rigidp}}$ for all $\mathtt{E}_j\in \Preds_{\rigidp} \setminus T$,
    \end{itemize}
    and, for every $\pred{E}_i\in T$, include in $\Precs$ some $\pi_{T,i}\in\Precs_o$ satisfying $\delta \in \phi_i^{\gamma_o(\pr_{T,i})}$ 
\end{enumerate}

\medskip

We have to show that if $\kstruct_{\rigidp}\models\phi$ then $\modelPR\models\phi$. If $\kstruct_{\rigidp}\models\phi$, then for all $\pr\in\Precs\subseteq\Precs_{o}$, $\kstruct_{\rigidp},\pr\models\phi$. 
We prove inductively on the structure of $\phi$ that, for $\pr\in\Precs$, that $\kstruct_{\rigidp},\pr,\varassign\models\phi$ iff $\modelPR,(\pi,f_{id}),\varassign\models\phi$.

\begin{description}
    \item[$\phi$ is of the form $\pred{P}(z)$]. Follows trivially

    \item[$\phi$ is of the form $\pred{P}(z,z')$]. Idem
    \item[$\phi$ is of the form $z \,\dot{=}\, z'$]. Idem
    \item[$\phi$ is of the form $\neg\psi$]. 
    $\kstruct_\rigidp,\!\pr,\!\varassign \models \neg\psi$ iff $\kstruct_\rigidp,\!\pr,\!\varassign\nvDash\psi$ (by the semantics), iff
    $\modelPR,(\pi,f_{id}),\varassign\nvDash\psi$ (by induction), iff 
    $\modelPR,(\pi,f_{id}),\varassign\models\neg\psi$ (by the semantics)
    \item[$\phi$ is of the form $\psi_1\land\psi_2$]. As easy as above
    \item[$\phi$ is of the form $\exists^{\lhd n}  x.\psi$]. For the forward direction, assume that $\kstruct_\rigidp,\!\pr,\!\varassign \models \exists^{\lhd n}  x.\psi$. Let $\Delta^{+}=\{ \delta\in\Delta \mid \kstruct_\rigidp,\!\pr,\!\varassign_{\set{x\mapsto\de}}\models\psi \}$ and $\Delta^{-}=|\{ \delta\in\Delta \mid \kstruct_\rigidp,\!\pr,\!\varassign_{\set{x\mapsto\de}}\nvDash\psi\}$. Then $\Delta=\Delta^{+}\cup\Delta^{-}$ and $|\Delta^{+}| \lhd n$ by the semantics. For all $\delta\in\Delta^{-}$ we have that $\modelPR,(\pi,f_{id}),\!\varassign_{\set{x\mapsto \de}}\nvDash\psi$ by induction. Hence, $|\Delta^{+}|=|\{ \delta \mid \modelPR,(\pi,f_{id}),\!\varassign_{\set{x\mapsto \de}}\models\psi \} | \lhd n$. For the backwards direction the argument can be reversed.
    \item[$\phi$ is of the form $\allstandd\psi$ and $\phi\in Dia_{\phi}\setminus FreeDia_{\phi}$]. 
    For the forward direction, assume that $\kstruct_\rigidp,\!\pr,\!\varassign \models \allstandd\psi$. Then by construction (P2) and the semantics of $\kstruct_\rigidp$, there is a precisification $\pr'\in\Precs\subseteq\Precs_o$ such that $\kstruct_\rigidp,\pr',\varassign\models\psi$. By induction, we obtain $\modelPR,(\pr',f_{id}),\!\varassign\models \psi$ and thus $\modelPR,(\pi,f_{id}),\!\varassign \models \allstandd\psi$ as desired. 
    
    For the backwards direction, assume $\modelPR,(\pi,f_{id}),\!\varassign \models \allstandd\psi$. Then by the semantics $\modelPR,(\pr',f_{id}),\!\varassign' \models \psi$ for some $\pr'\in\Precs$. By induction we obtain $\kstruct_\rigidp,\pr',\varassign\models\psi$ and hence $\kstruct_\rigidp,\pr,\varassign\models\allstandd\psi$.
    
    \item[$\phi$ is of the form $\allstandd\psi_i$ and $\phi\in FreeDia_{\phi}$]. Assume $z \in \{x,y\}$ is the free variable of $\phi$ and $\varassign(z)=\delta$. Forward direction,
    
    \begin{enumerate} 
        \item Assume $\kstruct_\rigidp,\!\pr,\!\varassign \models \allstandd\psi_i$ 
        \item By the construction of $\kstruct_\rigidp$, there is one $\rigidp$-type $T$ such that $\rigidp_i\in T$ and for all $\rigidp\in\Preds_\rigidp$ we have $\rigidp\in T$ iff $\delta\in\rigidp^{\kstruct_{\rigidp}}$
        \item By 2 and the construction of $\kstruct$ (P3), there is some $\pr_{T,i}\in\Precs$ and $\delta'\in\Delta$ such that $\delta'\in\psi_i^{\gamma_o(\pr_{T,i})}$ and $\delta'\in\rigidp^{\kstruct_{\rigidp}}$ for all $\rigidp\in T$ 
        \item By 1 and 3, we have $\kstruct_\rigidp,\pr_{T,i},\!\varassign'\models \psi_i$ for $\varassign'(z)=\delta'$ 
        \item By 4 and induction, $\modelPR,(\pr_{T,i},f_{id}),\!\varassign' \models \psi_i$.
        \item By 3 and the construction of $\modelPR$, there is a permutation $f$ such that $f(\delta')=\delta$
        \item By 5, 6 and \Cref{lemma:permutations-model-the-same}, $\modelPR,(\pr_{T,i},f),\!\varassign \models \psi_i$
        \item By 7 and the semantics, $\modelPR,(\pr,f_{id}),\!\varassign\models \allstandd\psi_i$.
    \end{enumerate}
    For the backwards direction,
    \begin{enumerate}
        \item Assume $\modelPR,(\pr,f_{id}),\!\varassign\models \allstandd\psi_i$ 
        \item By 1, and the semantics, there are some $\pr'$ and $f$ such that $\modelPR,(\pr',f),\!\varassign\models \psi_i$. Let $f^{-1}(\delta)=\delta'$
        \item By 2 and \Cref{lemma:permutations-model-the-same}, $\modelPR,(\pr',f_{id}),\!\varassign'\models \psi_i$ for $\varassign'(z)=\delta'$ 
        \item By induction and 3, we obtain that $\kstruct_\rigidp,\pr',\!\varassign'\models\psi_i$, hence $\kstruct_\rigidp,\!\varassign'\models\allstandd\psi_i$
        \item By 4 and the construction of $\kstruct_\rigidp$, $\delta'\in\rigidp_{i}^{\kstruct_\rigidp}$, and by the unary base case, $\delta'\in\rigidp_{i}^{\modelPR}$
        \item By 5 and \Cref{def:perm-closure}, $\delta\in\rigidp_{i}^{\modelPR}$ since $f(\delta')=\delta$, and by the unary base case, $\delta\in\rigidp_{i}^{\kstruct_\rigidp}$
        \item By 6 and the construction of $\kstruct_\rigidp$, $\kstruct_\rigidp,\!\varassign\models\allstandd\psi_i$ 
    \end{enumerate}
\end{description}

Until here we have shown that if $\kstruct_{o}\models\phi$ then
$\kstruct_{\rigidp}\models\phi$, if $\kstruct_{\rigidp}\models\phi$, then for all $\pr\in\Precs\subseteq\Precs_{o}$, $\kstruct_{\rigidp},\pr\models\phi$, and for all $\pr\in\Precs$, if $\kstruct_{\rigidp},\pr\models\phi$ then $\modelPR,(\pi,f_{id}),\varassign\models\phi$. Finally, by \Cref{lemma:permutations-model-the-same} we have that for all $\pr\in\Precs$ and $f\in\mathbb{P}_{\rigidp}$, if $\modelPR,(\pi,f_{id}),\varassign\models\phi$ then $\modelPR,(\pi,f),f\circ\varassign\models\phi$, therefore concluding that if $\kstruct_{o}\models\phi$ then $\modelPR\models\phi$.

\end{proof}

\begin{theorem-repeat}[of \Cref{lemma:stacked-model-satisfies-stack-formula}]
    The stacked model $\mathcal{I}^{\kstruct}$ of a first-order standpoint structure $\kstruct$ satisfies $\phi^m_\mathrm{stack}$.
\end{theorem-repeat}

\begin{proof}
    We show that the stacked model $\mathcal{I}^{\kstruct}$ of $\kstruct=\tuple{\Delta,\Precs,\sigma,\gamma}$, with $|\Precs|=2^m$ and signature $\tuple{\Preds,\emptyset,\{*\}}$, satisfies $\phi^m_\mathrm{stack}$ by showing the satisfaction of each of its conjuncts (F1)-(F5).

\begin{description}
    \item[ \ref{def-enum:everything-but-last-has-next}]  $\forall x.(\bigvee_{0\leq j< m} \neg \pred{L}_j(x)) \to \exists^{=1}y.\pred{F}(x,y)$. 
    Recall \Cref{def:stacked-model}.\ref{def-enum:stacked-Ls}: $\pred{L}_j^\mathcal{I} = \{ (\delta,i) \mid binary(i)[j]=1 \}$ . By \ref{def-enum:stacked-Ls}, the antecedent is satisfied for all assignments $x\to(\delta,i)$ with $i<2^m-1$. Then, recall  \Cref{def:stacked-model}.\ref{def-enum:stacked-F}: $\pred{F}^\mathcal{I} = \{ ((\delta,i),(\delta,{i+1})) \mid \delta \in \Delta,\ 0\leq i < 2^m-1 \}$. By \ref{def-enum:stacked-F}, the consequent is satisfied for all assignments $x\to(\delta,i)$ with $i<2^m-1$ as required.

    \item[\ref{def-enum:last-has-no-next}] $\forall x.(\bigwedge_{0\leq j< m} \pred{L}_j(x)) \to \exists^{=0}y.\pred{F}(x,y)$. By \ref{def-enum:stacked-Ls}, the antecedent is satisfied for all assignments $x\to(\delta,i)$ with $i=2^m-1$. And by \ref{def-enum:stacked-F}, there is no $y\to(\delta',k)$ such that $((\delta,2^m-1),(\delta',k))\in(\pred{F}(x,y))^{\mathcal{I}^{\kstruct}}$, so the consequent is satisfied as required.

    \item[ \ref{def-enum:everything-but-first-has-previous}] $\forall x.(\bigvee_{0\leq j< m} \pred{L}_j(x)) \to \exists^{=1}y.\pred{F}(y,x)$. By \ref{def-enum:stacked-Ls}, the antecedent is satisfied for all assignments $x\to(\delta,i)$ with $i>0$, in which case by \ref{def-enum:stacked-F} there exists exactly one assignment $y\to(\delta,{i-1})$ such that $\pred{F}(y,x)$, thus the consequent is satisfied as required.

    \item[ \ref{def-enum:first-has-no-previous}] $\forall x.(\bigwedge_{0\leq j< m} \neg\pred{L}_j(x)) \to \exists^{=0}y.\pred{F}(y,x)$. By \ref{def-enum:stacked-Ls}, the antecedent is satisfied for all assignments $x\to(\delta,i)$ with  $i=0$, in which case by \ref{def-enum:stacked-F} there is no assignment $y\to(\delta',k)$ such that $((\delta',k),(\delta,0))\in(\pred{F}(y,x))^{\mathcal{I}^{\kstruct}}$, so the consequent is satisfied as required.

    \item[ \ref{def-enum:level-counter}] First, by \ref{def-enum:stacked-F}, the antecedent is satisfied for all assignments $x\to(\delta,i)$ and $y\to(\delta,i+1)$ such that $\delta\in\Delta$ and $i<2^m-1$. 
    Then, notice that for all $0 \leq j < m$, if for any $0\leq j' < j$ we have $binary(i)[j']=0$, then we have $binary(i+1)[j]=binary(i)[j]$ (since the first of these previous $0$s of $i$ would have been the one to flip in $i+1$, leaving the preceding part of the encodings equal). And the opposite direction, if $binary(i+1)[j]\neq binary(i)[j]$, then all preceding positions $j'$ of $i$ must be $1$s. Given the construction of $\pred{L}_j$ by \ref{def-enum:stacked-Ls}, it is easy to see that the consequent encodes this implication, and hence it satisfies all assignments $x\to(\delta,i)$ and $y\to(\delta,i+1)$ as required. 
    
    \item[ \ref{def-enum:binary-at-same-level}] $\forall x. \forall y. \pred{P}(x,y) \to \bigwedge_{{0\leq j< m}}\pred{L}_j(x)\leftrightarrow\pred{L}_j(y)$ for all binary $\pred{P} \in \Preds$. By \ref{def-enum:stacked-binary-preds}, the antecedent can only be satisfied for assignments $x\to(\delta,i)$ and $y\to(\delta',k)$ such that $i=k$. Then we have $(\delta,i)\in\pred{L}_j^{\mathcal{I}^{\kstruct}}$ iff $binary(i)[j]=1$ iff $(\delta',i)\in\pred{L}_j^{\mathcal{I}^{\kstruct}}$ by \ref{def-enum:stacked-Ls}. Hence the consequent is satisfied for all assignments $x\to(\delta,i)$ and $y\to(\delta',i)$ as required.
\end{description}
\end{proof}

\begin{lemma}\label[lemma]{claim:stacked-bijection}
    Let $\mathcal{I}$ be a model of $\phi^m_\mathrm{stack}$. Then, there is a stacked interpretation $\mathcal{I}^{\kstruct}$ of some first-order standpoint structure $\kstruct$ with $2^m$ precisifications and a bijective function $\stacked:\Delta'\to\Delta \times \{0,\ldots,2^m-1\}$ such that for $\stacked(\delta'_1)=(\delta_1, i)$ and $\stacked(\delta'_2)=(\delta_2, k)$,
    \begin{enumerate}[ label={(C\arabic*)}, ref={(C\arabic*)}]
        \item $binary(i)[j]=1$ iff $\delta'\in\pred{L}_j^{\mathcal{I}}$ for $0\leq j < m$.\label{claim-enum:num}
        \item $\delta_1=\delta_2$ iff $(\delta'_1,\delta'_2)\in(\pred{F}^{\mathcal{I}})^+\cup(({\pred{F}}^{-1})^{\mathcal{I}})^+$, with $(\pred{F}^{\mathcal{I}})^+$ and $(({\pred{F}}^{-1})^{\mathcal{I}})^+$ the transitive closures of $\pred{F}^{\mathcal{I}}$ and $({\pred{F}}^{-1})^{\mathcal{I}}$. \label{claim-enum:dom}
        \item $\delta_1=\delta_2$ and $j=i+1$ iff $(\delta'_1,\delta'_2)\in\pred{F}^{\mathcal{I}}$. \label{claim-enum:F}
        \item $\pred{P}^\Inter = \bigcup_{0 \leq i < 2^m} \pred{P}^{\gamma(\pi_i)} \times \{i\}$ \ for all unary $\pred{P} \in \Preds$,  \label{def-enum:stacked-unary-preds}
        \item $\pred{P}^\Inter =  \{((\delta_1,i),(\delta_2,i)) \mid 0 \,{\leq}\, i \,{<}\, 2^m, (\delta_1,\delta_2) \,{\in}\, \pred{P}^{\gamma(\pi_i)} \}$ for all binary $\pred{P} \in \Preds$. \defend \label{def-enum:stacked-binary-preds}

    \end{enumerate}
\end{lemma}

\begin{proof}
    Let $\mathcal{I}= (\Delta', \cdot^\mathcal{I})$ be a first-order interpretation over the signature $\tuple{\Preds \cup \{\pred{F},\pred{L}_0, \ldots, \pred{L}_{m-1}\},\emptyset}$. Let $\kstruct=\tuple{\Dom,\Precs,\sigma,\gamma}$ be the standpoint structure over the signature $\tuple{\Preds,\emptyset,\{*\}}$ with $\Precs = \{\pi_0,\pi_1,\ldots,\pi_{2^m-1}\}$, $|\Delta|$ the number of disconnected components of $\pred{F}^{\mathcal{I}}$, and for all $\pi_i\in\Precs$, $\pred{P}^{\gamma(\pi_i)}=\{\dom(\delta')\mid\delta'\in\pred{P}^{\mathcal{I}},\num(\delta')=i\}$ for all unary $\pred{P} \in \Preds$, and $\pred{P}^{\gamma(\pi_i)}=\{(\dom(\delta_1'),\dom(\delta_2'))\mid(\delta_1',\delta_2')\in\pred{P}^{\mathcal{I}},\num(\delta_1')=i\}$ for all binary $\pred{P} \in \Preds$, where:
    \begin{itemize}
        \item $\num:\Delta'\rightarrow\{0,\ldots,2^m-1\}$ is such that $\num(\delta')=i$ if $binary(i)[j]=1$ iff $\delta'\in\pred{L}_j^{\mathcal{I}}$ for $0\leq j < m$. 
        \item $\dom:\Delta'\to\Delta$ is such that for $\delta'_1,\delta'_2\in\Delta'$ we have $\dom(\delta'_1)=\dom(\delta'_2)$ iff $(\delta'_1,\delta'_2)\in(\pred{F}^{\mathcal{I}})^+\cup(({\pred{F}}^{-1})^{\mathcal{I}})^+$. 
        \item $\stacked:\Delta'\to\Delta \times \{0,\ldots,2^m-1\}$ is $\stacked(\delta')=(\dom(\delta'),\num(\delta'))$.
        \item $\pred{P}^\Inter = \bigcup_{0 \leq i < 2^m} \pred{P}^{\gamma(\pi_i)} \times \{i\}$ \ for all unary $\pred{P} \in \Preds$,  \label{def-enum:stacked-unary-preds}
\item $\pred{P}^\Inter =  
\{((\delta_1,i),(\delta_2,i)) \mid 0 \,{\leq}\, i \,{<}\, 2^m, (\delta_1,\delta_2) \,{\in}\, \pred{P}^{\gamma(\pi_i)} \}
$ for all binary $\pred{P} \in \Preds$. \defend \label{def-enum:stacked-binary-preds}

    \end{itemize}
\smallskip

We let $\mathcal{I}^{\kstruct}$ be the stacked interpretation of $\kstruct$ and proceed to show that the conditions (C1-5) hold.

By the definition, it is clear that \ref{claim-enum:num} and \ref{claim-enum:dom} are met. It remains to prove that $\stacked$ is a bijection and \ref{claim-enum:F} is satisfied. For this, we will show that if $\stacked(\delta')=(\delta,i)$ then $\delta'$ is the $i$th element of a $\delta$-chain of length $2^m$, i.e., we show by induction on $i$ that $\delta'$ has at most one $\pred{F}$-successor and $\pred{F}$-predecessor satisfying \ref{claim-enum:F}, and exactly $2^m-1-i$ $(\pred{F}^+)$-successors and $i$ $(\pred{F}^+)$-predecessors.

\begin{itemize}[leftmargin=1em,label=-]
    \item $\num(\delta')=i$ has at most one $\pred{F}$-successor and exactly $2^m-1-i$ $(\pred{F}^+)$-successors:
    \begin{description}[leftmargin=2em,font=\normalfont\underline, style=standard]
        \item[Case: $i=2^m-1$.] By the definition of $\num$ and axiom \ref{def-enum:last-has-no-next}, there is no $\delta''\in\Delta'$ such that $(\delta',\delta'')\in\pred{F}^{\mathcal{I}}$, hence $\delta'$ has 0 $\pred{F}$-successors (satisfying \ref{claim-enum:F}) and 0 $(\pred{F}^+)$-successors as desired. 
        \item[Case: $i<2^m-1$.] By the definition of $\num$ and axiom \ref{def-enum:everything-but-last-has-next}, there is exactly one element $\delta''\in\Delta'$ such that $(\delta',\delta'')\in\pred{F}^{\mathcal{I}}$. Thus, $\delta'$ one $\pred{F}$-successor. Moreover, $\dom(\delta'')=\delta$ and from axiom \ref{def-enum:level-counter} and the definition of $\num$, we obtain that $\num(\delta'')=i+1$, thus satisfying \ref{claim-enum:F}. By induction, $\delta''$ has $2^m-1-(i+1)$ $(\pred{F}^+)$-successors and thus $\delta'$ has $2^m-1-i$ $(\pred{F}^+)$-successors as required.
    \end{description}
    \item $\num(\delta')=i$ has at most one $\pred{F}$-predecessor and exactly $i$ $(\pred{F}^+)$-predecessors:
    \begin{description}[font=\normalfont\underline, style=standard]
        \item[Case $0=i$] By the definition of $\num$ and axiom \ref{def-enum:first-has-no-previous}, there is no $\delta''\in\Delta'$ such that $(\delta'',\delta')\in\pred{F}^{\mathcal{I}}$, hence $\delta'$ has 0 $\pred{F}$-predecessors (satisfying \ref{claim-enum:F}) and $(\pred{F}^+)$-predecessors as desired.
        \item[Case $0<i$] By the definition of $\num$ and axiom \ref{def-enum:everything-but-first-has-previous} there is exactly one element $\delta''\in\Delta'$ such that $(\delta'',\delta')\in\pred{F}^{\mathcal{I}}$. Thus, $\delta'$ has one $\pred{F}$-predecessor. Moreover, $\dom(\delta'')=\delta$ and, from axiom \ref{def-enum:level-counter} and the definition of $\num$, we obtain that $\num(\delta'')=i-1$, thus satisfying \ref{claim-enum:F}. Thus, by induction, $\delta''$ has $i-1$ $(\pred{F}^+)$-predecessors and $\delta'$ has $i$ $(\pred{F}^+)$-predecessors as required.
    \end{description}
\end{itemize}

We now show that the function $\stacked$ is bijective, first proving injection and then surjection.
For injection, for the sake of contradiction, suppose that $\stacked(\delta')=(\delta,i)$, $\stacked(\delta'')=(\delta,i)$ and $\delta'_1\neq\delta'_2$. From the argument above, both $\delta'$ and $\delta''$ are each the (single) $i$th element of a different $\delta$-chain of length $2^m$. But, from the definition of $\dom$, we have that if $\dom(\delta')=\dom(\delta'')$ then $(\delta',\delta'')\in(\pred{F}^{\mathcal{I}})^+\cup(({\pred{F}}^{-1})^{\mathcal{I}})^+$, thus they would need to belong to the same $\delta$-chain, thus leading to a contradiction.

For surjection, for the sake of contradiction, suppose that for some $\delta\in\Delta$ and $i\in\{0,\dots,2^m-1\}$ there is no $\delta'\in\Delta'$ such that $\stacked(\delta')=(\delta,i)$. Recall that we set $|\Delta|$ to be the number of disconnected components of $(\pred{F}^{\mathcal{I}})^+$ and by construction the function $\dom$ maps the elements of each component to a distinct domain element in $\Delta$. Thus, there must be some $\delta''\in\Delta'$ such that $\stacked(\delta'')=(\delta,k)$. Then, by the argument above, $\delta''$ is the $k$th element of a $\delta$-chain, in which the $i$th element is some $\delta'$ such that $\stacked(\delta')=(\delta,i)$, thus leading to a contradiction.
\end{proof}

\begin{theorem-repeat}[of \Cref{lemma:model-satisfies-stack-formula-if-isomorphic-to-stacked-model}]
A first-order interpretation $\mathcal{I}$ satisfies $\phi^m_\mathrm{stack}$ if and only if it is isomorphic to a stacked interpretation $\mathcal{I}^{\kstruct}$ of some first-order standpoint structure $\kstruct$ with $2^m$ precisifications.
\end{theorem-repeat}

\begin{proof}
Let $\mathcal{I}= (\Delta', \cdot^\mathcal{I})$ be a first-order interpretation over the signature $\tuple{\Preds \cup \{\pred{F},\pred{L}_0, \ldots, \pred{L}_{m-1}\},\emptyset}$. Let $\mathcal{I}^{\kstruct}$ be the stacked interpretation of $\kstruct=\tuple{\Dom,\Precs,\sigma,\gamma}$ with $\Precs = \{\pi_0,\pi_1,\ldots,\pi_{2^m-1}\}$ and $|\Delta|$ the number of disconnected components of $(\pred{F}^{\mathcal{I}})^+$. 

\smallskip

First, if $\mathcal{I}$ is isomporphic to a stacked interpretation $\mathcal{I}^{\kstruct}$ then it satisfies $\phi^m_\mathrm{stack}$ by Lemma \ref{lemma:stacked-model-satisfies-stack-formula}.
It remains to prove the other direction, i.e., that if $\mathcal{I}$ satisfies $\phi^m_\mathrm{stack}$, then it is isomorphic to $\mathcal{I}^{\kstruct}$. Assume that $\mathcal{I}$ satisfies $\phi^m_\mathrm{stack}$. Thus:
\begin{description}
    \item[\ref{def-enum:stacked-domain}] By \Cref{claim:stacked-bijection} there is a bijective function $\stacked:\Delta'\to\Delta \times \{0,\ldots,2^m-1\}$, therefore $\Delta'$ is isomorphic to $\Delta \times \{0,\ldots,2^m-1\}$ as required.
    \item[\ref{def-enum:stacked-Ls}] From \Cref{claim:stacked-bijection}.\ref{claim-enum:num} we obtain directly that $\pred{L}_j^\mathcal{I} = \{ \delta' \mid binary(i)[j]=1, \stacked(\delta')=(\delta,i) \}$ as required.
    \item[\ref{def-enum:stacked-F}]  From \Cref{claim:stacked-bijection}.\ref{claim-enum:F} we obtain directly that $\pred{F}^\mathcal{I} = \{ (\delta'_1,\delta'_2) \mid \delta'_1,\delta'_2\in \Delta', \stacked(\delta'_1)=(\delta,i),\stacked(\delta'_2)=(\delta,i+1)\}$ as required
    \item[\ref{def-enum:stacked-unary-preds}] Trivial, since we can let $\kstruct$ be such that $\pred{P}^{\gamma(\pi_i)}=\{\delta\mid\delta'\in\pred{P}^\Inter, \stacked(\delta')=(\delta,i)\}$ as required.
    \item[\ref{def-enum:stacked-binary-preds}] Let $\stacked(\delta'_1)=(\delta_1,i)$ and $\stacked(\delta'_2)=(\delta_2,k)$. If $(\delta'_1,\delta'_2)\in\pred{P}^{\mathcal{I}}$, then by \ref{def-enum:binary-at-same-level} we have that $\delta'_1\in\pred{L}_j^{\mathcal{I}}$ iff $\delta'_2\in\pred{L}_j^{\mathcal{I}}$ for all $\pred{L}_j\in\{\pred{L}_0,\dots,\pred{L}_m\}$. Therefore, by \Cref{claim:stacked-bijection}.\ref{claim-enum:num}, $i=k$. Thus, we can let $\pred{P}^{\gamma(\pi_i)}=\{(\delta_1,\delta_2)\mid(\delta'_1,\delta'_2)\in\pred{P}^\Inter, \stacked(\delta'_1)=(\delta_1,i),\stacked(\delta'_2)=(\delta_2,i)\}$ as required. 
\end{description}

\noindent We conclude that $\mathcal{I}$ is isomorphic to a stacked model $\mathcal{I}^{\kstruct}$ of $\kstruct$ as required.

    
\end{proof}

\begin{theorem-repeat}[of \Cref{lemma:permutational-closure-iff-stacked-model}]
Let $\phi$ be a frugal \SmonCtwo sentence over the signature $\tuple{\Preds,\emptyset,\{*\}}$. Let $\kstruct=\tuple{\Dom,\Precs,\sigma,\gamma}$ be a standpoint structure for the signature $\tuple{\Preds \uplus \Preds_{\rigidp},\emptyset,\{*\}}$, with all predicates from $\Preds_{\rigidp}=\{\rigidp_0,\ldots,\rigidp_\ell\}$ rigid, and $|\Precs| = 2^m$
. Then,
$$\modelPR \models \phi \Longleftrightarrow \mathcal{I}_\kstruct \models \mathsf{Trans}(\phi).$$ 
\end{theorem-repeat}

\begin{proof}

We show that, for $\pr_i\in\Precs$, $f\in\mathbb{P}_{\rigidp}$ and an assignment $\varassign$, we have $\modelPR
,(\pr_i,f),\varassign\models\phi$ iff $\mathcal{I}_\kstruct,\varassign' \models \mathsf{tr}(\phi)$ where $\varassign'(z)=(f(\varassign(z)),i)$ for $z\in\{x,y\}$. And, by virtue of Lemma \ref{lemma:permutations-model-the-same}, we can focus on showing that 
$$\modelPR
,(\pr_i,f_{id}),\varassign\models\phi \quad\text{iff}\quad \mathcal{I}_\kstruct,\varassign' \models \mathsf{tr}(\phi)$$ 

where $\varassign'(z)=(\varassign(z),i)$ for all $z\in\{x,y\}$. We show this inductively on the structure of $\phi$



\begin{description}
    \item[Base case: $\phi$ is of the form $\pred{P}(z)$]. First, we have $\modelPR,(\pr_i,f_{id}),\varassign\models\pred{P}(z)$ iff $\varassign(z)\in\pred{P}^{\gamma'((\pr_i,f_{id}))}$ by the semantics. Then $\varassign(z)\in\pred{P}^{\gamma'((\pr_i,f_{id}))}$ iff $\varassign(z)\in\pred{P}^{\gamma(\pr_i)}$ by the construction of the permutational closure (Def. \ref{def:perm-closure}). Finally, $\varassign(z)\in\pred{P}^{\gamma(\pr_i)}$ iff $(\varassign(z),i)\in\pred{P}^{\mathcal{I}}$ from the construction of the stacked model (Def. \ref{def:stacked-model}), and $(\varassign(z),i)\in\pred{P}^{\mathcal{I}}$ iff $\mathcal{I}_\kstruct,\varassign'\models \pred{P}(z)$ from the semantics again.

    \item[Base case: $\phi$ is of the form $\pred{P}(z,z')$]. 
    \item[Base case: $\phi$ is of the form $z=z'$]. 
    \item[Case: $\phi$ is of the form $\neg\psi$].
    \item[Case: $\phi$ is of the form $\psi_1\land\psi_2$]. 
    \item[Case: $\phi$ is of the form $\exists^{\lhd n}  z.\psi$]. 
    Forward direction: 
    \begin{enumerate}
        \item Assume $\modelPR,(\pr_i,f_{id}),\varassign\models\exists^{\lhd n}  z.\psi$, thus $|\{ \delta \mid \modelPR,\!(\pr_i,f_{id}),\!\varassign_{\set{z\mapsto\de}}\models\psi \}|\lhd n$
        \item For the sake of contradiction, assume that $\mathcal{I}_\kstruct,\varassign' \nvDash \exists^{\lhd n} z. (\phi^=_\mathtt{L}(x,y) \wedge  \mathsf{tr}(\psi))$
        \item From 2 and the semantics, 
        $| \{ (\de',k) \mid \mathcal{I}_\kstruct,\varassign'_{\set{z\mapsto(\de',k)}}\models \phi^=_\mathtt{L}(x,y)\}\cap\{ (\de',k) \mid \mathcal{I}_\kstruct,\varassign'_{\set{z\mapsto(\de',k)}}\models \mathsf{tr}(\psi)) \}|\not\!\!\lhd\ n$
        \item From the construction of $\mathcal{I}_\kstruct$ (\ref{def-enum:stacked-Ls}, Def. \ref{def:stacked-model}), and given that $\varassign'(x)=(\varassign(x),i)$ and $\varassign'(y)=(\varassign(y),i)$, then for all $(\de',k)$ such that $\mathcal{I}_\kstruct,\varassign'_{\set{z\mapsto(\de',k)}}\models \phi^=_\mathtt{L}(x,y)$ we have $k=i$
        \item From 3 and 4, $|\{(\de',i)\in\Delta' \mid \mathcal{I}_\kstruct,\varassign'_{\set{z\mapsto(\de',i)}}\models \mathsf{tr}(\psi) \}|\not\!\!\lhd\  n$
        \item From the inductive hypothesis, for each $(\de',i)$ such that $\mathcal{I}_\kstruct,\varassign'_{\set{z\mapsto(\de',i)}}\models \mathsf{tr}(\psi)$ then $\modelPR,(\pr_i,f_{id}),\varassign_{\set{z\mapsto\de'}}\models\psi$
        \item From 5 and 6 we obtain that $|\{ \delta \mid \modelPR,\!(\pr_i,f_{id}),\!\varassign_{\set{z\mapsto\de}}\models\psi \} |\not\!\!\lhd\  n$, thus reaching a contradiction with 1
    \end{enumerate}
    Backwards direction: 
    \begin{enumerate}
        \item Assume $\mathcal{I}_\kstruct,\varassign' \models \exists^{\lhd n} z. (\phi^=_\mathtt{L}(x,y) \wedge  \mathsf{tr}(\psi))$
        \item From 1 and the semantics, 
        $| \{ (\de',k) \mid \mathcal{I}_\kstruct,\varassign'_{\set{z\mapsto(\de',k)}}\models \phi^=_\mathtt{L}(x,y)\}\cap\{ (\de',k)\mid \mathcal{I}_\kstruct,\varassign'_{\set{z\mapsto(\de',k)}}\models \mathsf{tr}(\psi)\}|\lhd n$
     
   \item From the construction of $\mathcal{I}_\kstruct$ (\ref{def-enum:stacked-Ls}, Def. \ref{def:stacked-model}), and given that $\varassign'(x)=(\varassign(x),i)$ and $\varassign'(y)=(\varassign(y),i)$, then for all $(\de',k)$ such that $\mathcal{I}_\kstruct,\varassign'_{\set{z\mapsto(\de',k)}}\models \phi^=_\mathtt{L}(x,y)$ we have $k=i$
        \item From 2 and 3, $|\{(\de',i)\in\Delta' \mid \mathcal{I}_\kstruct,\varassign'_{\set{z\mapsto(\de',i)}}\models \mathsf{tr}(\psi) \}|\lhd n$
        
        \item For the sake of contradiction, assume $|\{ \delta \mid \modelPR,\!(\pr_i,f_{id}),\!\varassign_{\set{z\mapsto\de}}\models\psi \}|\not\!\!\lhd\  n$
        \item From the inductive hypothesis, for each $\delta$ such that $\modelPR,\!(\pr_i,f_{id}),\!\varassign_{\set{z\mapsto\de}}\models\psi$, then $\mathcal{I}_\kstruct,\varassign'_{\set{z\mapsto(\delta,i)}}\models\mathsf{tr}(\psi)$
        \item From 5 and 6 we have that $\{ (\de',i)\mid \mathcal{I}_\kstruct,\varassign'_{\set{z\mapsto(\de',i)}}\models \mathsf{tr}(\psi)\}|\not\!\!\lhd\  n$ thus reaching a contradiction with \textcolor{blue}{4}

    \end{enumerate}

    \item[Case: $\phi$ is of the form $\allstandd\psi$]. 
    For the forward direction. 
    \begin{enumerate}
        \item Assume $\modelPR,(\pr_i,f_{id}),\{z_\mathrm{mf}\mapsto\delta',z_\mathrm{nf}\mapsto\delta''\}\models\allstandd\psi$ 
        \item From 1, $\modelPR,(\pr_i,f_{id}),\{z_\mathrm{mf}\mapsto\delta'\}\models\allstandd\psi$ since only $z_\mathrm{mf}$ could be a free variable in $\psi$
        \item From 2, there is some $(\pr_j,f^{-1})$, $\modelPR,(\pr_j,f^{-1}),\{z_\mathrm{mf}\mapsto\delta'\}\models\psi$ 
        \item From 3 and Lemma \ref{lemma:permutations-model-the-same}, $\modelPR,(\pr_j,f_{id}),\{z_\mathrm{mf}\mapsto f(\delta')\}\models\psi$ 
        \item From 4 and the inductive hypothesis, then $\mathcal{I}_\kstruct,\{z_\mathrm{mf}\mapsto (f(\delta'),j)\} \models \mathsf{tr}(\psi)$ 
        \item Let $\Preds_{\psi}:=\{\pred{E}\in\Preds_{\rigidp} \mid \delta'\in\pred{E}^{\modelPR}\}$.  From the construction of $\modelPR$ (Def.\ref{def:perm-closure}), $\Preds_{\psi}=\{\pred{E}\in\Preds_{\rigidp}\mid f(\delta')\in\pred{E}^{\modelPR}\}$
        \item From 6, the construction of $\mathcal{I}_\kstruct$ (\ref{def-enum:stacked-unary-preds}, Def.\ref{def:stacked-model}) and the fact that $\rigidp\in\Preds_{\rigidp}$ in $\kstruct$ is rigid, $(\delta',k),(f(\delta'),k)\in\rigidp^{\mathcal{I}}$
         iff $\rigidp\in\Preds_{\psi}$ for all $0\leq k <2^{m}$
        \item From 5 and 7, $\mathcal{I}_\kstruct,\{z_\mathrm{mf}\mapsto (f(\delta'),j),z_\mathrm{nf}\mapsto (\delta',i)\} \models \phi^=_\rigidp(x,y) \wedge \mathsf{tr}(\psi)$
        \item From 8, $\mathcal{I}_\kstruct,\{z_\mathrm{nf}\mapsto (\delta',i)\} \models  \exists z_{{\mathrm{mf}}}.\phi^=_\rigidp(x,y) \wedge \mathsf{tr}(\psi)$
        \item From 9, $\mathcal{I}_\kstruct,\{z_\mathrm{mf}\mapsto (\delta',i)\} \models  \forall z_\mathrm{nf}.x\,\dot{=}\,y \to  \exists z_{{\mathrm{mf}}}.\phi^=_\rigidp(x,y) \wedge \mathsf{tr}(\psi)$
    \end{enumerate}
    For the backwards direction 
    \begin{enumerate}
        \item Assume $\mathcal{I}_\kstruct,\{z_\mathrm{mf}\mapsto (\delta',i)\}\models  \forall z_\mathrm{nf}.x\,\dot{=}\,y \to  \exists z_{\mathrm{\mathrm{mf}}}.\phi^=_\rigidp(x,y) \wedge \mathsf{tr}(\psi)$
        \item From 1, $\mathcal{I}_\kstruct,\{z_\mathrm{nf}\mapsto (\delta',i)\}\models \exists z_{{\mathrm{mf}}}.\phi^=_\rigidp(x,y) \wedge \mathsf{tr}(\psi)$
        \item From 2, there is some $(\delta,j)$ such that  $\mathcal{I}_\kstruct,\{z_\mathrm{mf}\!\mapsto\!(\delta,j),z_\mathrm{nf}\!\mapsto\!(\delta',i)\}\models \phi^=_\rigidp(x,y) \wedge \mathsf{tr}(\psi)$
        \item Let $\Preds_{\psi}:=\{\pred{E}\in\Preds_{\rigidp}|(\delta,j)\in\pred{E}^{\mathcal{I}}\}$. From 3, the construction of $\mathcal{I}_\kstruct$ (\ref{def-enum:stacked-unary-preds}, Def.\ref{def:stacked-model}) and the fact that $\rigidp\in\Preds_{\rigidp}$ in $\kstruct$ is rigid, $\delta,\delta'\in\rigidp^{\gamma'}$ for $\rigidp\in\Preds_{\psi}$.
        \item From 3 and the semantics $\mathcal{I}_\kstruct,\{z_\mathrm{mf}\!\mapsto\!(\delta,j)\}\models \mathsf{tr}(\psi)$ since $z_\mathrm{mf}$ is the only free variable in $\mathsf{tr}(\psi)$
        \item From 5 and the inductive hypothesis, $\modelPR,(\pr_j,f_{id}),\{z_\mathrm{mf}\!\mapsto\!\delta\}\models\psi$
        \item From 4 and 6, there is a permutation such that $f(\delta)=\delta'$, thus $\modelPR,(\pr_j,f),\{z_\mathrm{mf}\!\mapsto\!\delta'\}\models\psi$
        \item From 7 and the semantics, $\modelPR,(\pr_i,f_{id}),\{z_\mathrm{mf}\!\mapsto\!\delta'\}\models\allstandd\psi$
    \end{enumerate}

Now, toward the statement of the Lemma, assume $\modelPR \models \phi$. Then, $\modelPR
,(\pr_i,f),\varassign\models\phi$ holds for all $\pr_i\in\Precs$, $f\in\mathbb{P}_{\rigidp}$ and assignments $\varassign$. Then, $\mathcal{I}_\kstruct,\varassign' \models \mathsf{tr}(\phi)$ where $\varassign'(z)=(\varassign(z),i)$, thus $\mathcal{I}_\kstruct,\varassign' \models \mathsf{tr}(\phi)$ for all $\varassign'$ where $\varassign'(x)$ and $\varassign'(y)$ coincide on their second component. Therefore, we obtain $\mathcal{I}_\kstruct \models \mathsf{Trans}(\phi)$. 

For the converse direction, assume $\mathcal{I}_\kstruct \models \mathsf{Trans}(\phi)$. Then, $\mathcal{I}_\kstruct \models \forall x.\forall y.(x\,\dot{=}\,y \to \mathsf{tr}(\phi))$. Thus, for all valuations such that $\varassign'(x)=\varassign'(y)$, $\mathcal{I}_\kstruct,\varassign'\models \mathsf{tr}(\phi)$. Then, $\modelPR,(\pr_i,f_{id}),\varassign\models\phi$ for all $\pr_i\in\Precs$ and all valuations such that $\varassign(x)=\varassign(y)$. And, since $\phi$ has no free variables, this implies that $\modelPR,(\pr_i,f_{id})\models\phi$ for all $\pr_i\in\Precs$. Finally, by \Cref{lemma:permutations-model-the-same}, $\modelPR,(\pr_i,f)\models\phi$ for all $(\pr_i,f)\in\Precs'$ and thus $\modelPR\models\phi$ as desired.
\qedhere
\end{description}
\end{proof}

\begin{theorem-repeat}[of \Cref{lemma:permutational-closure-iff-rigid}]
Let $\kstruct=\tuple{\Dom,\Precs,\sigma,\gamma}$ be a \foss for the signature $\tuple{\Preds \uplus \Preds_{\rigidp},\emptyset,\{*\}}$. Then, all predicates from $\Preds_{\rigidp}=\{\rigidp_0,\ldots,\rigidp_\ell\}$ are rigid iff
$ \mathcal{I}_\kstruct \models \phi^\ell_{rig\rigidp}$. 
\end{theorem-repeat}

\begin{proof}
For the forward direction. If all predicates from $\Preds_{\rigidp}=\{\rigidp_0,\ldots,\rigidp_\ell\}$ are rigid, then by \Cref{def:stacked-model}.\ref{def-enum:stacked-unary-preds}, $\pred{E}^\Inter = \bigcup_{0 \leq i < 2^m} \pred{E}^{\kstruct} \times \{0,\ldots,2^m-1\}$ for all $\pred{E} \in \Preds_{\rigidp}$. Then, consider $\phi^\ell_{rig\rigidp}$. By \Cref{def:stacked-model}.\ref{def-enum:stacked-F}, the antecedent is satisfied for all assignments $x\to(\delta,i)$ and $y\to(\delta,i+1)$ such that $\delta\in\Delta$ and $i<2^m-1$. Then it is clear that  $(\delta,i)\in\pred{E}^\Inter$ iff $(\delta,i+1)\in\pred{E}^\Inter$ for all $\pred{E} \in \Preds_{\rigidp}$, thus satisfying the consequent. Therefore, $ \mathcal{I}_\kstruct \models \phi^\ell_{rig\rigidp}$ as desired.

For the backwards direction. Assume that $ \mathcal{I}_\kstruct \models \phi^\ell_{rig\rigidp}$. Then for all assignments $x\to(\delta,i)$ and $y\to(\delta,i+1)$ such that $\delta\in\Delta$ and $i<2^m-1$, we have  $(\delta,i)\in\pred{E}^\Inter$ iff $(\delta,i+1)\in\pred{E}^\Inter$ for all $\pred{E} \in \Preds_{\rigidp}$. Thus by the construction of $\mathcal{I}_\kstruct$, for all $\pr_i\in\Precs$ with $i<2^m-1$, we have $\delta\in\pred{E}^{\gamma(\pr_i)}$ iff $\delta\in\pred{E}^{\gamma(\pr_{i+1})}$. Therefore we obtain $\pred{E}^{\gamma(\pr_i)}=\pred{E}^{\gamma(\pr_j)}$ for all $i,j\in\{0,\ldots,2^m-1\}$ and thus the predicates in $\Preds_{\rigidp}$ are rigid.
\end{proof}

\section{Converting Monodic Stand\-point \texorpdfstring{$\mathcal{SROIQB}_s$}{SROIQbs} / \texorpdfstring{$\mathcal{SHOIQB}_s$}{SHOIQbs} to Monodic Stand\-point \texorpdfstring{$\mathcal{ALCHOIQB}_s$}{ALCHOIQBs}}\label{sec:SSSROIQ}

This section provides the details on the logics $\mathbb{S}^\mathrm{mon}_{\mathcal{SROIQB}_s}$
and $\mathbb{S}^\mathrm{mon}_{\mathcal{SROIQB}_s}$ mentioned in \Cref{sec:addingrolechain} as well as the claimed back-translation into $\mathbb{S}^\mathrm{mon}_{\mathcal{ALCOIQB}^\mathsf{Self}}$ announced therein as well.

$\mathbb{S}^\mathrm{mon}_{\mathcal{SROIQB}_s}$ adds the feature of monodic stand\-point-aware modelling to $\mathcal{SROIQB}_s$, a DL obtained from the well-known DL $\mathcal{SROIQ}$~\cite{DBLP:conf/kr/HorrocksKS06} by allowing Boolean role expressions over simple roles~\cite{DBLP:conf/jelia/RudolphKH08}. Let us reiterate that the $\mathcal{SROIQ}$ family serves as the logical foundation of popular ontology languages like OWL~2~DL.

We will next introduce the syntax $\mathbb{S}^\mathrm{mon}_{\mathcal{SROIQB}_s}$. Some  specifics closely follow previous work on the more restrictive sentential fragment of standpoint $\mathcal{SROIQ}b_s$ by  \citeauthor{sententialFOLandOWL} (\citeyear{sententialFOLandOWL}), in particular regarding some design decisions concerning how to take into account the global syntactic constraints of $\mathcal{SROIQB}_s$ in a standpointified setting. As opposed to the very minimalistic syntax definition in the main part of the paper, our syntactic definition will also explicitly include some constructors that could be considered as ``syntactic sugar'', but are convenient to have available as ``first-class citizens'' when establishing normal forms.

It will become apparent that $\mathbb{S}^\mathrm{mon}_{\mathcal{SROIQB}_s}$ is an extension of $\mathbb{S}^\mathrm{mon}_{\mathcal{ALCOIQB}^\mathsf{Self}}$.

We again start from a signature $\tuple{\Preds, \Consts, \Stands}$ where $\Preds$ only contains unary and binary predicates. We find it convenient to subdivide $\Preds$ according to the predicate arity into $\Preds_1$ and $\Preds_2$ and refer to them as \emph{concept names} and \emph{role names}, respectively. $\Preds_2$ is subdivided further into \emph{simple role names} $\Preds^\mathrm{s}_2$ and \emph{non-simple role names} $\Preds^\mathrm{ns}_2$, the latter being strictly ordered by some strict partial order $\prec$.

Then, the set $\mathbf{E}^\mathrm{smpl}_\mathrm{rol}$ of \emph{simple role expressions} is defined by
\begin{narrowalign}
	$\rolexpR,\rolexpR' \ebnfeq \rolS \mid \rolS^- \mid  \neg \rolexpR \mid \rolexpR\cap\rolexpR' \mid \rolexpR\cup\rolexpR'$,
\end{narrowalign}
with $\rolS {\,\in\,} \Preds^\mathrm{s}_2$, while the set of (arbitrary) \emph{role expressions} is $\mathbf{E}_\mathrm{rol} = \mathbf{E}^\mathrm{smpl}_\mathrm{rol} \cup \Preds^\mathrm{ns}_2 \cup \{\rolR^- \mid \rolR \in \Preds^\mathrm{ns}_2 \}$. The order $\prec$ is then extended to $\mathbf{E}_\mathrm{rol}$ by making all elements of $\mathbf{E}^\mathrm{smpl}_\mathrm{rol}$ $\prec$-minimal and stipulating $\rolR^- \prec \rolexpR$ iff $\rolR \prec \rolexpR$ for all $\rolR \in \Preds^\mathrm{ns}_2$ and $\rolexpR \in \mathbf{E}_\mathrm{rol}$, and likewise $\rolexpR \prec \rolR^-$ iff $\rolexpR \prec \rolR$.
\emph{Concept expressions} $\mathbf{E}_\mathrm{con}$ are defined via %
	{
		\begin{align*}
			\conC,\conD \ebnfeq \conA \mid \{o\} \mid \top \mid \bot \mid \neg\conC \mid \conC\sqcap\conD \mid \conC\sqcup\conD \mid \forall\rolexpR.\conC \mid \exists\rolexpR.\conC \mid \exists\rolexpR'.\mathit{Self} \mid \atmost{n}\rolexpR'.C \mid \atleast{n}\rolexpR'.C
            \mid \standb{\ste} C \mid \standd{\ste} C,
		\end{align*}
	}%
with \mbox{$\conA\in \Preds_1$}, \mbox{$o\in \Consts$}, \mbox{$\rolexpR \in \mathbf{E}_\mathrm{rol}$}, \mbox{$\rolexpR' \in  \mathbf{E}^\mathrm{smpl}_\mathrm{rol}$}, \mbox{$n\in \mathbb{N}$}, and $\ste\in \StandExps$ (see \Cref{def:FOSLsyntax}).
We note that any concept expression $C$ can be put in negation normal form, denoted $\mathsf{NNF}_\mathrm{con}(C)$, where concept negation $\neg$ only occurs in front of concept names, nominals, or $\mathit{Self}$ concepts.

As before, a \emph{general concept inclusion} (GCI) is an expression of the form $C \sqsubseteq D$ with $C,D \in \mathbf{E}_\mathrm{con}$. A \emph{role chain axiom}, also referred to as \emph{(complex) role inclusion axiom} (RIA), is an expression of one of the following forms:
\begin{align*}
\rolexpR_1 \circ \ldots \circ \rolexpR_n \sqsubseteq \rolR\\
\rolexpR_1 \circ \ldots \circ \rolexpR_n \circ \rolR \sqsubseteq \rolR\\
\rolR \circ \rolexpR_1 \circ \ldots \circ \rolexpR_n \sqsubseteq \rolR\\
\rolR \circ\rolR \sqsubseteq \rolR, \!\!                                   
\end{align*}
where $\rolR \in \Preds^\mathrm{ns}_2$, while $\rolexpR_i \in \mathbf{E}_\mathrm{rol}$ and $\rolexpR_i \prec \rolR$ for all $i\in \{1,\ldots,n\}$. We refer to the set of all GCIs and RIAs as \emph{axioms} and denote it with $\mathbf{Ax}$.                                          Finally the set of $\mathbb{S}^\mathrm{mon}_{\mathcal{SROIQB}_s}$ sentences is defined by (letting $\alpha \in \mathbf{Ax}$)

$$ \phi, \psi ::= \alpha \mid \neg \phi \mid \phi \wedge \psi \mid \phi \vee \psi \mid \standb{\ste} \phi \mid \standd{\ste} \phi.$$

We let $\mathbb{S}^\mathrm{mon}_{\mathcal{SHOIQB}_s}$ denote the set of set of $\mathbb{S}^\mathrm{mon}_{\mathcal{SROIQB}_s}$ sentences wherein every occurring RIA is of the shape $\rolS \sqsubseteq \rolR$ or $\rolR \circ\rolR \sqsubseteq \rolR$.

We obtain the semantics of $\mathbb{S}^\mathrm{mon}_{\mathcal{SROIQB}_s}$ by extending our translation into FOSL to RIAs as follows: any RIA of the form $\rolexpS_1 \circ \ldots \circ \rolexpS_k \sqsubseteq \rolR$ occurring inside the to-be-translated $\mathbb{S}^\mathrm{mon}_{\mathcal{SROIQB}_s}$ sentence $\phi$ is replaced by 
$$
\forall x_0,\ldots,x_k. \Big(\bigwedge_{0 < i \leq k}\mathsf{rtrans}(x_{i-1},x_{i},\rolexpS_i)\Big) \to \rolR(x_0,x_k).
$$

We see that, unless $k=1$, the obtained formula is not in \SmonCtwo anymore, it is still a FOSL formula, thus the translation-based semantics is still well defined. However, unlike for $\mathbb{S}^\mathrm{mon}_{\mathcal{ALCOIQB}^\mathsf{Self}}$, this extended translation cannot serve as an immediate tool to establish decidability, let alone tight complexity bounds.

However, it is possible to harness existing techniques for eliminating RIAs from $\mathcal{SROIQ}$ ontologies \cite{Kazakov08,DemriN05}. This does, however, require a bit of extra care as RIAs may hold or fail to hold precisification-wise.

In the following, we describe a multi-step transformation process that, as a whole, takes a 
$\mathbb{S}^\mathrm{mon}_{\mathcal{SROIQB}_s}$ sentence as an input and returns an equisatisfiable $\mathbb{S}^\mathrm{mon}_{\mathcal{ALCOIQB}^\mathsf{Self}}$ sentence of possibly exponential size. For the subclass of  $\mathbb{S}^\mathrm{mon}_{\mathcal{SHOIQB}_s}$
sentences, the resulting sentence is even guaranteed to be of polynomial size.

\subsection{Negation Normal Form}

Given an arbitrary $\mathbb{S}^\mathrm{mon}_{\mathcal{SROIQB}_s}$ formula, we use the following recursively defined function $\mathsf{NNF}$ to transform it into \emph{negation normal form}:

\begin{align*}
\mathsf{NNF}(\neg (\phi \wedge \psi)) & = \mathsf{NNF}(\neg\phi) \vee \mathsf{NNF}(\neg\psi)\\
\mathsf{NNF}(\phi \wedge \psi) & = \mathsf{NNF}(\phi) \vee \mathsf{NNF}(\psi)\\
\mathsf{NNF}(\neg (\phi \vee \psi)) & = \mathsf{NNF}(\neg\phi) \wedge \mathsf{NNF}(\neg\psi)\\
\mathsf{NNF}(\phi \vee \psi) & = \mathsf{NNF}(\phi) \wedge \mathsf{NNF}(\psi)\\
\mathsf{NNF}(\neg \standb{\ste}\phi) & = \standd{\ste}\mathsf{NNF}(\neg\phi)\\
\mathsf{NNF}( \standb{\ste}\phi) & = \standb{\ste}\mathsf{NNF}(\phi)\\
\mathsf{NNF}(\neg \standd{\ste}\phi) & = \standb{\ste}\mathsf{NNF}(\neg\phi)\\
\mathsf{NNF}( \standd{\ste}\phi) & = \standd{\ste}\mathsf{NNF}(\phi)\\
\mathsf{NNF}(\neg ( C \sqsubseteq D ) ) & = \top \sqsubseteq \exists \rolU.\mathsf{NNF}_\mathrm{con}(C \sqcap \neg D)\\
\mathsf{NNF}( C \sqsubseteq D ) & = \top \sqsubseteq \mathsf{NNF}_\mathrm{con}(\neg C \sqcup D)\\
\mathsf{NNF}(\neg \rho) & = (\top \sqsubseteq {\leqslant}1 \pred{F}_\rho.\top) \wedge 
 ( \{o\} \sqsubseteq \exists \pred{F}_\rho.\exists \rolexpS_1. \ldots \exists \rolexpS_k.\forall \rolR^-.\forall \pred{F}_\rho^-.\neg\{o\}) \qquad \mbox{for any RIA } \rho \ = \ \rolexpS_1 \circ \ldots \circ \rolexpS_k \sqsubseteq \rolR\\
\mathsf{NNF}( \rho ) & = \rho 
\end{align*}
where $\rolU$ stands for the \emph{universal role}, which can be written as $\rolR \cup \neg\rolR$ for an arbitrary $\rolR \in \mathbf{P}_2^\mathrm{s}$, while $o$ is an arbitrary constant and $\pred{F}_\rho$ is a fresh role (freshly introduced for every RIA $\rho$). It is routine to verify, that the transformation is polytime and preserves satisfiability.

\subsection{Separation of RIAs}

For the next step, we assume we are given a $\mathbb{S}^\mathrm{mon}_{\mathcal{SROIQB}_s}$ sentence $\phi$ in negation normal form. Our goal is to obtain an equisatisfiable $\mathbb{S}^\mathrm{mon}_{\mathcal{SROIQB}_s}$ sentence of the form $\phi_\mathrm{RIA} \wedge \phi_\mathrm{rest}$, where $\phi_\mathrm{RIA}$ is a conjunction of RIAs while $\phi_\mathrm{rest}$ is a $\mathbb{S}^\mathrm{mon}_{\mathcal{SROIQB}_s}$ sentence in negation normal form without occurrences of RIAs.
To arrive at this special form, we employ a trick by \citeauthor{sententialFOLandOWL} (\citeyear{sententialFOLandOWL}) that allows to take RIAs out of their boolean and modal contexts but endows them with a ``switch'', so they can be activated or deactivated from inside such contexts (for more details and a more comprehensive discussion, we refer the reader to the provided literature). 

First, for every $\rolR\in\Preds^\mathrm{ns}_2$, we introduce a copy $\underline{\rolR}$. Moreover, we introduce a simple role name $\rolS_\rho$ for each RIA $\rho$ inside $\phi$.
Thereby, the non-simple role names inherit their ordering $\prec$ from $\Preds^\mathrm{ns}_2$ and we also let $\underline{\rolR} \prec \rolR$ for each $\rolR{\,\in\,}\Preds^\mathrm{ns}_2$.

Then, we let $\phi_\mathrm{rest}$ be obtained from $\phi$ by
\begin{itemize}
\item replacing every occurring RIA $\rho$ by the GCI $\top \sqsubseteq \exists \rolS_\rho.\mathsf{Self}$ and
\item replacing every $\exists \rolR$ for non-simple $\rolR$ with $\exists \underline{\rolR}$.
\end{itemize}

We obtain $\phi_\mathrm{RIA}$ as the conjunction over the set of RIAs consisting of
\begin{itemize}
\item the RIA
$\underline{\rolR}  \sqsubseteq  \rolR$
for every $\rolR \in \Preds^\mathrm{ns}_2$ and 
\item for every RIA $\rho$ inside $\phi$, the RIA $BG(\rho)$, with $BG$ defined by
	{ \begin{align*}
			\rolexpR_1 \circ   \ldots  \circ  \rolexpR_n  \sqsubseteq  \rolR                      \mapsto \hspace{1ex} & \rolS_\rho  \circ  \rolexpR_1  \circ   \ldots  \circ  \rolexpR_n   \sqsubseteq  \underline{\rolR}                                      &
			 \rolexpR_1 \circ   \ldots  \circ  \rolexpR_n \circ  \rolR  \sqsubseteq  \rolR  \mapsto \hspace{1ex} & \rolS_\rho   \circ  \rolexpR_1  \circ   \ldots  \circ  \rolexpR_n  \circ  \underline{\rolR}   \sqsubseteq  \underline{\rolR}    \\
			\rolR  \circ  \rolexpR_1 \circ   \ldots  \circ  \rolexpR_n  \sqsubseteq  \rolR  \mapsto \hspace{1ex} & \underline{\rolR}   \circ  \rolexpR_1  \circ   \ldots  \circ  \rolexpR_n  \circ  \rolS_\rho   \sqsubseteq  \underline{\rolR}  &
			\rolR  \circ  \rolR  \sqsubseteq  \rolR                                      \mapsto \hspace{1ex} & \rolS_\rho   \circ   \underline{\rolR}  \circ   {\rolR}   \sqsubseteq \rolR ,\!\!
		\end{align*}}%
\end{itemize}

Again it is easy to see that the translation produces a $\mathbb{S}^\mathrm{mon}_{\mathcal{SROIQB}_s}$ sentence of the announced shape and can be computed in polytime (hence the output sentence is of polynomial size with respect to the input sentence).

What remains to be argued is equisatisfiability of $\phi$ and $\phi_\mathrm{RIA} \wedge \phi_\mathrm{rest}$. To this end we show how, given a model for one formula, a model of the other can be constructed.

We first observe that $\phi_\mathrm{RIA} \wedge \phi_\mathrm{rest} \models \phi$, that is, every model of 
$\phi_\mathrm{RIA} \wedge \phi_\mathrm{rest}$ readily serves as a model of $\phi$.

What remains to be shown is that every model of $\phi$ gives rise to a model of $\phi_\mathrm{RIA} \wedge \phi_\mathrm{rest}$. We argue that any model $\kstruct=\tuple{\Dom, \Precs, \sigma, \gamma}$ can be turned into a model of $\phi_\mathrm{RIA} \wedge \phi_\mathrm{rest}$ by appropriately extending the interpretations to the newly introduced predicates. This is obtained by, for any $\pr \in \Precs$ letting:
\begin{itemize}
\item $\underline{\rolR}^{\gamma(\pr)} = \rolR^{\gamma(\pr)}$ for all $\rolR \in \Preds^\mathrm{ns}_2 $
\item $(\rolS_\rho)^{\gamma(\pr)} = \{(\delta,\delta) \mid \delta \in \Delta\}$ if $\rho$ is satisfied by ${\gamma(\pr)}$ and $(\rolS_\rho)^{\gamma(\pr)} = \emptyset$ otherwise.
\end{itemize}

\subsection{Compiling RIAs into Concept Expressions}

Through the previous step, we obtained a sentence, where all occurring RIAs are conjunctively combined and stipulated to universally hold across all precisifications. Also, by means of the syntactic restrictions, the set of these RIAs form what is commonly known as a ``regular RBox''.
Based on earlier work by \citeauthor{DemriN05} (\citeyear{DemriN05}), \citeauthor{Kazakov08} (\citeyear{Kazakov08}) showed that a regular RBox can be eliminated from a $\SROIQ$ knowledge base by compiling the RIAs into the GCIs. This is essentially done by tending to every occurring concept expression $\forall \rolR. C$ for nonsimple $\rolR$ and, by means of automata-based techniques and auxiliary concept symbols, making sure that $\forall \rolexpS_1. \ldots \forall \rolexpS_k. C$ also holds for all role expression sequences $\rolexpS_1 \ldots \rolexpS_k$ that would, by means of (possibly iterated) RIA applications give rise to an $\rolR$-connection. As this compilation can be executed locally (i.e., axiom-wise) for every occurrence of $\forall \rolR. C$, the technique straightforwardly applies to our setting, where GCIs occur inside boolean formulae and modal operators in $\phi_\mathrm{rest}$. Kazakov's compilation will generally translate a single GCI $\alpha$ into several GCIs $\alpha_1,...\alpha_k$, so inside $\phi_\mathrm{rest}$, we will locally replace the occurrence of $\alpha$ by $\alpha_1 \wedge \ldots \wedge \alpha_k$   
The result obtained, say $\phi_\mathrm{final}$, is a rewritten $\phi_\mathrm{rest}$, while the  $\phi_\mathrm{RIA}$ part can be discarded from the sentence. Then it is clear that the sentence thus obtained is indeed a $\mathbb{S}^\mathrm{mon}_{\mathcal{ALCOIQB}^\mathsf{Self}}$ sentence as intended. Its size is bounded exponentially by the size of $\phi_\mathrm{RIA} \wedge \phi_\mathrm{rest}$ in general while the bound is polynomial if the original sentence was in $\mathbb{S}^\mathrm{mon}_{\mathcal{SHOIQB}_s}$.

\end{document}